\documentclass[12pt]{article}
\def\citep#1{\cite{#1}}

\newif\ifpreprint
\preprinttrue

\def\preprintstart{}
\def\preprintstop{}
\usepackage{fullpage}

\usepackage[utf8]{inputenc} 
\usepackage[T1]{fontenc}    

\usepackage{amsfonts}       
\usepackage{amssymb}
\usepackage{amsmath}
\usepackage{euscript} 
\usepackage{stmaryrd} 
\usepackage{color}

\usepackage[colors]{optsys}

\newcommand{\GenericSeparation}[1]%
        {\mathrel{\raisebox{-0.2em}{$\substack{\text{\small $\Vert$} \\[-1.7mm] \line(1,0){10}\\ #1}$}}}

\newcommand{\Tsep}{\GenericSeparation{t}}

\renewcommand{\AGENT}{{\mathbb A}}

\newcommand{\AgentSubsetW}{W}

\newcommand{\Precedence}{\ParentRelation}

\renewcommand{\Precedence}{\EDGE}
\renewcommand{\AGENT}{\VERTEX}

\renewcommand{\agent}{\alpha}
\renewcommand{\agentbis}{\gamma}
\renewcommand{\bgent}{\gamma}
\renewcommand{\Bgent}{\Gamma}
\renewcommand{\cgent}{\lambda}
\renewcommand{\Cgent}{\Lambda}

\renewcommand{\relation}{\mathcal{R}}
\renewcommand{\relationter}{\relationbis}

\newcommand{\UpperIndex}{\AgentSubsetW}
\newcommand{\upperAgentSubsetW}{^{\text{\tiny $\UpperIndex$}}}
\newcommand{\MinusupperAgentSubsetW}{^{\text{\tiny$-\UpperIndex$}}}

\newcommand{\ParentalPrecedence}{\Precedence\upperAgentSubsetW}

\newcommand{\ConverseParentalPrecedence}{\Precedence\MinusupperAgentSubsetW}

\newcommand{\tinyTC}[4]{{#1^{\text{\tiny $#2 #3 #4$}}}}
\newcommand{\TransitiveClosureRelation}[2]{\tinyTC{#1}{}{#2}{+}}
\newcommand{\TransitiveClosureConverseRelation}[2]{\tinyTC{#1}{-}{#2}{+}}
\newcommand{\TransitiveReflexiveClosureRelation}[2]{\tinyTC{#1}{}{#2}{*}}
\newcommand{\TransitiveReflexiveClosureConverseRelation}[2]{\tinyTC{#1}{-}{#2}{*}}

\newcommand{\TopologicalRelation}{\relationter\upperAgentSubsetW}
\newcommand{\ConverseTopologicalRelation}{\relationter\MinusupperAgentSubsetW}

\newcommand{\TransitiveClosureParentalPrecedence}%
{\TransitiveClosureRelation{\Precedence}{\UpperIndex}}
\newcommand{\TransitiveClosureConverseParentalPrecedence}%
{\TransitiveClosureConverseRelation{\Precedence}{\UpperIndex}}
\newcommand{\TransitiveReflexiveClosureParentalPrecedence}%
{\TransitiveReflexiveClosureRelation{\Precedence}{\UpperIndex}}
\newcommand{\TransitiveReflexiveClosureConverseParentalPrecedence}%
{\TransitiveReflexiveClosureConverseRelation{\Precedence}{\UpperIndex}}

\newcommand{\ConditionalAncestor}{\TransitiveReflexiveClosureParentalPrecedence}
\newcommand{\ConverseConditionalAncestor}{\TransitiveReflexiveClosureConverseParentalPrecedence}

\newcommand{\ACTIVE}{\mathcal{A}}
\newcommand{\ConditionalActive}{\ACTIVE\upperAgentSubsetW}

\newcommand{\ConditionalAscendent}{\mathcal{B}\upperAgentSubsetW}
\newcommand{\ConverseConditionalAscendent}{\mathcal{B}\MinusupperAgentSubsetW}
\newcommand{\ConditionalCommonCause}{\mathcal{K}\upperAgentSubsetW}
\newcommand{\ConditionalActiveTwo}{\ConditionalCommonCause}
\newcommand{\ConverseConditionalCommonCause}{\ConditionalCommonCause}
\newcommand{\ConverseConditionalActiveTwo}{\ConverseConditionalCommonCause}
\newcommand{\Cousinhood}{\mathcal{C}\upperAgentSubsetW}

\newcommand{\ConditionalActiveThree}{%
  \Bp{
    \bp{\ConditionalAscendent \cup \ConditionalActiveTwo}
    \Cousinhood
    \bp{\ConverseConditionalAscendent \cup \ConverseConditionalActiveTwo}
  }
}

\newcommand{\ConditionalActiveNew}{\ConditionalActive}

\newcommand{\topology}{{\cal T}} 
\newcommand{\ClopenSet}{F}
\newcommand{\ClosedSet}{C}

\newcommand{\OpenSet}{O}

\newcommand{\TopologyLower}[1]{\topology_{{\text{\tiny $#1$}}}}

\newcommand{\TopologicalClosure}[2]{\overline{#2}^{\text{\tiny $#1$}}}

\newcommand{\OrientedGraph}{\VERTEX,\EDGE}
\newcommand{\npOrientedGraph}{\np{\OrientedGraph}}

\ifdefined\path\renewcommand{\path}{\varrho}
\else\newcommand{\path}{\varrho}
\fi

\newcommand{\RelTheta}{\Theta\upperAgentSubsetW}
\newcommand{\ConverseRelTheta}{\Theta\MinusupperAgentSubsetW}
\newcommand{\ConditionalActiveThreeTheta}{\np{\RelTheta\Cousinhood\ConverseRelTheta}}

\def\ConditionalUp{\npTransitiveReflexiveClosure{\Converse{\Precedence}\Delta_{\Complementary\AgentSubsetW}} \Converse{\Precedence}}
\def\ConditionalDown{{\Precedence} \npTransitiveReflexiveClosure{\Delta_{\Complementary\AgentSubsetW}\Precedence}}

\usepackage{amsmath}

\title{Topological Conditional Separation}

\author{%
  Michel De Lara$^\dagger$,
  Jean-Philippe Chancelier\footnote{CERMICS, Ecole des Ponts, Marne-la-Vall\'ee, France},
  Benjamin Heymann\footnote{Criteo AI Lab, Paris, France}
}

\date{\today}

\begin{document}

\maketitle

\begin{abstract}

  Pearl's d-separation is a foundational notion to study conditional
  independence between random variables.
  We define the topological conditional separation and we show that it
  is equivalent to the d-separation, extended beyond acyclic graphs, be they
  finite or infinite.
\end{abstract}

\section{Introduction}
As the world shifts toward more and more data-driven decision-making,
causal inference is taking more space in applied sciences, statistics and machine learning. 
This is because it allows for  better, more robust decision-making, and
provides a way to interpret the data that goes beyond correlation
\citep{pearl2018book}. 
In his seminal work~\citep{pearl1995causal}, Pearl  builds on graphical
models~\citep{cowell2006probabilistic} to propose the so-called
do-calculus, and he notably introduces the notion of d-separation on a
Directed Acyclic Graph (DAG).

This paper has two companion papers~\cite{Chancelier-De-Lara-Heymann-2021,Heymann-De-Lara-Chancelier-2021},
Altogether, the three of them aim at providing
another perspective on conditional independence and do-calculus. 
In this paper, we consider directed graphs (DGs), not necessarily acyclic, and
we introduce a suitable topology on the set of vertices. Then, we define 
the new notion of topological conditional separation on DGs, and we 
prove its equivalence with an extension of Pearl's d-separation on DGs.
The topological separation is practical because it just requires to check that
two sets are disjoint. By contrast, the d-separation requires to check that
\emph{all} the paths that connect two variables are blocked. 
Moreover, as its name suggests, the topological separation has a
theoretical interpretation which motivates a detour by the theory of
Alexandrov topologies.

The paper is organized as follows.
In Sect.~\ref{Alexandrov_topology_on_a_graph},
we provide background on binary relations and graphs, and then
we present Alexandrov topologies induced by binary relations.
In Sect.~\ref{Equivalence_between_d-separation_and_t-separation},
we recall the definition of d-separation,
then introduce a suitable topology on the set of vertices,
and define a new notion of conditional topological separation (t-separation).
Then, we show that d-separation and t-separation are equivalent,
and we put forward a practical characterization of t-separation between subsets of
vertices.
We provide additional material on Alexandrov topologies 
in Appendix~\ref{Additional_results_alexandrov}, and we 
relegate technical lemmas in Appendix~\ref{Additional_Lemmas}.

\section{Alexandrov topology on a graph}
\label{Alexandrov_topology_on_a_graph}

In~\S\ref{Background_on_binary_relations_and_graphs}, we provide background on
binary relations and graphs.
In~\S\ref{Topologies_induced_by_binary_relations},
we present Alexandrov topologies induced by binary relations.

\subsection{Background on binary relations, graphs and topologies}
\label{Background_on_binary_relations_and_graphs}

We use the notation \( \ic{r,s}=\na{r,r+1,\ldots,s-1,s} \) for any two integers
$r$, $s$ such that $r \leq s$.

\subsubsection{Binary relations}
\label{Binary_relations}

Let $\AGENT$ be a nonempty set (finite or not). 
We recall that a \emph{(binary) relation}~$\relation$ on~$\AGENT$ is
a subset $\relation \subset \AGENT\times\AGENT $ and that 
\( \bgent\, \relation\, \cgent \) means 
\( \np{\bgent,\cgent} \in \relation \).
For any subset \( \Bgent \subset \AGENT \), 
the \emph{(sub)diagonal relation} is \( \Delta_{\Bgent} = \bset{ \np{\bgent,\cgent} \in \AGENT\times\AGENT }%
{ \bgent=\cgent \in \Bgent } \)
and the \emph{diagonal relation} is \( \Delta=\Delta_{\AGENT} \).
A \emph{foreset} of a relation~$\relation$ is
any set of the form \( \relation \, \cgent = 
\defset{ \bgent \in  \AGENT }{ \bgent\, \relation \, \cgent } \),
where \( \cgent \in \AGENT \), 
or, by extension, of the form \( \relation \, \Cgent = 
\defset{ \bgent \in  \AGENT }{ \exists \cgent \in \Cgent \eqsepv \bgent\,
  \relation \, \cgent } \), where \( \Cgent \subset \AGENT \).
An \emph{afterset} of a relation~$\relation$ is
any set of the form \( \bgent \, \relation = 
\defset{ \cgent \in  \AGENT }{ \bgent\, \relation \, \cgent } \),
where \( \bgent \in \AGENT \), 
or, by extension, of the form \( \Bgent \, \relation = 
\defset{ \cgent \in  \AGENT }{ \exists \bgent \in \Bgent \eqsepv \bgent\,
  \relation \, \cgent } \), where \( \Bgent \subset \AGENT \).
The \emph{opposite} or \emph{complementary~$\Complementary{\relation}$} of a binary
relation~$\relation$ is the relation~$\Complementary{\relation}=\AGENT\times\AGENT\setminus\relation$,
that is, defined by \( \bgent\, \relation^{\mathsf{c}} \, \cgent \iff 
\neg \np{ \bgent\, \relation \, \cgent } \).
The \emph{converse~$\Converse{\relation}$} of a binary relation~$\relation$ is
defined by \( \bgent\, \Converse{\relation} \, \cgent \iff \cgent\, \relation \, \bgent
\).
A relation~$\relation$ is \emph{symmetric} if \( \Converse{\relation}=\relation
\),
and is \emph{anti-symmetric} if \( \Converse{\relation} \cap \relation \subset
\Delta \). 

The \emph{composition} $\relation\relation'$ of two
binary relations~$\relation, \relation'$ on~$\AGENT$ is defined by
\( \bgent (\relation\relation') \cgent \iff
\exists \delta \in  \AGENT \), \( \bgent\, \relation \, \delta \)
and \( \delta\, \relation' \, \cgent \);
then, by induction we define
\( \relation^{n+1}=\relation\relation^{n} \) for \( n \in \NN^* \). 
The \emph{transitive closure} of a binary relation~$\relation$ is
\( \TransitiveClosure{\relation} = \cup_{k=1}^{\infty} \relation^{k} \)
(and $\relation$ is  \emph{transitive} if \( \TransitiveClosure{\relation}=\relation \))
and the \emph{reflexive and transitive closure} is 
\( \TransitiveReflexiveClosure{\relation}= \TransitiveClosure{\relation} \cup
\Delta = \cup_{k=0}^{\infty} \relation^{k} \) with the convention
$\relation^0=\Delta$.
A \emph{partial equivalence relation} is a symmetric and transitive binary
relation (generally denoted by~$\sim$ or~$\equiv$).
An \emph{equivalence relation} is a reflexive, symmetric and transitive binary
relation.

\subsubsection{Preorders}
\label{Preorders}

A \emph{preorder}  (or ``quasi-ordering'')
on~$\AGENT$ is a reflexive and transitive binary relation (generally
denoted by~$\moinsfine$),
whereas an \emph{order} is an anti-symmetric preorder (generally
denoted by~$\leq$).
For a preorder, the foreset (resp. afterset) of a subset \( \Bgent \subset \AGENT \) is called the
\emph{downset} (resp. \emph{upset}) of~$\Bgent$ and is denoted
by~\( \downarrow \! \Bgent \) (resp. by~\( \uparrow \! \Bgent \)):
\begin{equation*}
  \downarrow \!\Bgent
  = \bset{\agent\in\AGENT}{\exists\agentbis\in\Bgent \eqsepv
    \agent\moinsfine\agentbis}
  \eqsepv
  \uparrow \!\Bgent
  = \bset{\agent\in\AGENT}{\exists\agentbis\in\Bgent \eqsepv
    \agentbis\moinsfine\agent}
  \eqfinp     
\end{equation*}
Then, a subset \( \Bgent \subset \AGENT \) is called an \emph{upper set}
(resp. a \emph{lower set}) --- or also an \emph{upward closed set} (resp. \emph{downward
  closed  set}) --- with respect to the preorder~\( \moinsfine \) if 
\( \downarrow\!\Vertex  \, \subset \Vertex\)
(resp. \( \uparrow \!\Vertex \, \subset \Vertex\)) or, equivalently, if
\( \downarrow\!\Vertex  \, =\Vertex\)
(resp. \( \uparrow \!\Vertex \, =\Vertex\)).

\subsubsection{Graphs}

Let $\VERTEX$ be a nonempty set (finite or not), whose elements are called \emph{vertices}.  Let
\( \EDGE \subset \VERTEX\times\VERTEX \) be a relation on~$\VERTEX$, whose
elements are ordered pairs (that is, couples) of vertices called \emph{edges}.
The first element of an edge is the \emph{tail of the edge},
whereas the second one is the \emph{head of the edge}.
Both tail and head are called \emph{endpoints} of the edge,
and we say that the edge connects its endpoints.
We define a \emph{loop} 
as an element of \( \Delta \cap \EDGE \), that is, a loop is an edge that
connects a vertex to itself.

A \emph{graph}, as we use it throughout this paper, is a
couple~$(\VERTEX,\EDGE)$ where \( \EDGE \subset \VERTEX\times\VERTEX \).
This definition is very basic and we now stress proximities and differences with
classic notions in graph theory.
As we define a graph, it may hold a finite or infinite number of vertices;
there is at most one edge that has a couple of ordered vertices as single endpoints,
hence a graph (in our sense) is not a multigraph (in graph theory);
loops are not excluded (since we do not impose $\Delta \cap \EDGE=\emptyset$).
Hence, what we call a graph would be called a directed simple graph permitting
loops in graph theory.

\subsubsection{Topologies}

We refer the reader to \cite[Chapter~4]{Goubault-Larrecq:2013} for notions in topology.
Let $\VERTEX$ be a nonempty set.
The set~\( \topology \subset 2^\VERTEX \) is said to be a \emph{topology}
on~$\VERTEX$ if \( \topology \)  contains both \( \emptyset, \AGENT \) and is stable under 
the union and finite 
intersection operations.
The space \( \np{\VERTEX,\topology} \) is called \emph{topological space}.
Any element \( \OpenSet \in \topology \) is called an \emph{open set}
(more precisely a $\topology$-open set),
and any element in
\begin{equation}
  \topology'= \bset{ \ClosedSet \subset \VERTEX }%
  { \Complementary{\ClosedSet} \in \topology }
  \label{eq:Topology'}
\end{equation}
is called a \emph{closed set} (more precisely a $\topology$-closed set).
For any subset  $\Vertex \subset \VERTEX$, the intersection of all the closed
sets that contain~$\Vertex$ is a closed set called \emph{topological closure}
and denoted by~\( \TopologicalClosure{}{\Vertex} \) (or, when needed,
\( \TopologicalClosure{\topology}{\Vertex} \)). 

%
A \emph{clopen set} (more precisely a $\topology$-clopen set)
is a subset of~$\VERTEX$ which is both closed and open, that
is, an element of \( \topology \cap \topology' \).
A topological space~\( \np{\VERTEX,\topology} \) is said to be disconnected if it is the union of two
disjoint nonempty open sets; otherwise, it is said to be connected.
A subset \( \Vertex \subset \VERTEX \) of~$\VERTEX$ is said to be
\emph{connected} (more precisely $\topology$-connected)
if it is connected under
its \emph{subspace topology} \( \topology\cap\Vertex =
\bset{ \OpenSet\cap\Vertex \in \topology }{ \OpenSet \in \topology } \)
(also called \emph{trace topology} or \emph{relative topology}).
A \emph{connected component} of the topological space
\( \np{\VERTEX,\topology} \)  (also called a $\topology$-connected component)
is a maximal (for the inclusion order) connected subset.
A connected component is necessarily closed
and the connected components of \( \np{\VERTEX,\topology} \) form a partition
of~$\VERTEX$ \cite[Exercise~4.11.13]{Goubault-Larrecq:2013}.
Any clopen set is a union of (possibly infinitely many) connected components.

Let \( \np{\VERTEX_i,\topology_i} \), $i=1,2$ be two topological spaces.
The \emph{product topology} \( \topology_1\otimes\topology_2 \) is the smallest
subset \( \topology \subset 2^{\VERTEX_1\times\VERTEX_2} \)
which is a topology on the product set \( \VERTEX_1\times\VERTEX_2 \) 
and which contains all the finite rectangles
\( \bset{\OpenSet_1\times\OpenSet_2}{\OpenSet_i \in \topology_i, i=1,2} \).

\subsubsubsection{Specialization preorder}

With any topology~\( \topology \) on~$\VERTEX$, one associates the so-called 
\emph{specialization} (or canonical) \emph{preorder}
as the binary relation~\( \moinsfine_\topology\) on~$\VERTEX$ defined by
\cite[\S~4.2.1, Lemma~4.2.7]{Goubault-Larrecq:2013}
\begin{equation}
  \bgent \moinsfine_\topology \cgent \iff
  \bgent \in \TopologicalClosure{\topology}{\cgent}
  \qquad \bp{\forall  \bgent, \cgent \in\VERTEX }
  \eqfinp 
\end{equation}
The relation~\( \moinsfine_\topology\) is reflexive and transitive, hence is a
preorder (hence the notation).
Following the notation in~\S\ref{Binary_relations}
--- with the notation $\downarrow_\topology$ for a downset and $\uparrow_\topology$ for an upset --- 
we have that
\begin{equation}
  \downarrow_\topology \!\cgent =  \TopologicalClosure{\topology}{\cgent}
\eqsepv  \forall \cgent \in\VERTEX
  \eqfinv 
\end{equation}
%
it is readily shown (and well-known \cite[Lemmas~4.2.6 and~4.2.7]{Goubault-Larrecq:2013})
that every open set is an upper set 
and every closed set is a lower set. 

\subsubsubsection{Preorder topology}

It can be shown that, for any preorder~$\moinsfine$ on~$\VERTEX$, the set
\begin{equation}
  \TopologyLower{\moinsfine} = \bset{ \OpenSet \subset \VERTEX }%
  { \uparrow \! \OpenSet \subset \OpenSet }
  \label{eq:TopologyLowermoinsfine}
\end{equation}
is a topology and that it is the finest topology~\( \topology \) that has~$\moinsfine$
as specialization order (that is, such that \(
\moinsfine_{\topology}= \moinsfine \))
\cite[Proposition~4.2.11]{Goubault-Larrecq:2013}.
The topology~\( \TopologyLower{\moinsfine} \) is an Alexandrov topology
as follows.

\subsubsubsection{Alexandrov topology}

The set~\( \topology \subset 2^\VERTEX \) is said to be an \emph{Alexandrov topology}
on~$\VERTEX$ if \( \topology \) contains both \( \emptyset, \AGENT \) and is stable under 
the union and (not necessarily finite) intersection operations.
If \( \topology \) is an Alexandrov topology, then \( \topology' \) in~\eqref{eq:Topology'}
also is an Alexandrov topology, that we call the \emph{dual (Alexandrov)
  topology}~\cite{Bouacida-Echi-Salhi:1996}.

It is established that a topology~\( \topology \) is an Alexandrov topology if
and only if \( \topology =\TopologyLower{\moinsfine_\topology} \), where
\( \moinsfine_\topology \) is the specialization preorder of~\( \topology \),
that is, if and only if 
the open sets (in~$\topology$) are exactly
the upper sets (with respect to~\( \moinsfine_\topology \)) --- or, equivalently,
the closed sets (in~$\topology'$) are exactly
the lower sets (with respect to~\( \moinsfine_\topology \))
\cite[Proposition~4.2.11, Exercise~4.2.13]{Goubault-Larrecq:2013}.
Thus, if \( \topology \) is an Alexandrov topology, we have that 
\begin{equation}
\Bp{  \OpenSet \in \topology \iff 
  \downarrow_\topology \! \OpenSet  \, \subset \OpenSet }
  \mtext{ and }
\Bp{  \ClosedSet \in \topology' \iff  
  \uparrow_\topology \! \ClosedSet \, \subset \ClosedSet }
  \eqfinp 
\end{equation}


In an Alexandrov topology, it can be shown that,
for any family \( \sequence{\Bgent_\scenario}{\scenario\in\SCENARIO} \)
of subsets~$\Bgent_\scenario \subset \VERTEX$, 
the topological closure satisfies
\begin{equation}
  \TopologicalClosure{}{\bigcup_{\scenario\in\SCENARIO}\Bgent_\scenario} =
  \bigcup_{\scenario\in\SCENARIO}\TopologicalClosure{}{\Bgent_\scenario}
  \eqsepv \forall \Bgent_\scenario \subset \VERTEX
  \eqsepv \scenario\in\SCENARIO
  \eqfinp 
  \label{eq:TopologicalClosureAlexandrovTopology}
\end{equation}
Indeed, as \( \Bgent_\scenario \subset
\TopologicalClosure{}{\bigcup_{\scenariobis\in\SCENARIO}\Bgent_{\scenariobis}}
\),
we get that
\( \bigcup_{\scenario\in\SCENARIO}\Bgent_\scenario
\subset \bigcup_{\scenario\in\SCENARIO}\TopologicalClosure{}{\Bgent_\scenario}
\subset
\TopologicalClosure{}{\bigcup_{\scenario\in\SCENARIO}\Bgent_\scenario} \).
As the set \(
\bigcup_{\scenario\in\SCENARIO}\TopologicalClosure{}{\Bgent_\scenario} \)
is closed, by definition of Alexandrov topology, we conclude.
By~\eqref{eq:TopologicalClosureAlexandrovTopology}, it is readily deduced that
\begin{equation}
  \TopologicalClosure{}{\Bgent} = \, \downarrow_\topology \! \Bgent
  \eqsepv \forall \Bgent \subset \VERTEX
  \eqfinp 
  \label{eq:TopologicalClosureAlexandrovTopology=LowerSet}
\end{equation}
%

\subsection{Alexandrov topology induced by a binary relation}
\label{Topologies_induced_by_binary_relations}

As recalled, a topology is an Alexandrov topology if
and only if it is the topology of a preorder 
\cite[Exercise~4.2.13]{Goubault-Larrecq:2013}.
In fact, one can associate a topology with any binary relation
(see~\eqref{eq:TopologyLowerEDGE} below)
and prove that a topology is an Alexandrov topology if
and only if it is the topology of a binary relation 
~\cite[Théorème~1.2]{Bouacida-Echi-Salhi:1996}.
In Proposition~\ref{pr:TopologyLowerEDGE}, we analyze the Alexandrov topology
induced by a binary relation;
we recover known results \cite[Théorème~1.2]{Bouacida-Echi-Salhi:1996}
and we add some new results.
Additional results are provided in Appendix~\ref{Additional_results_alexandrov}.

\begin{proposition}
  \label{pr:TopologyLowerEDGE}
  Let $(\VERTEX,\EDGE)$ be a graph, that is,
$\VERTEX$ is a set and \( \EDGE \subset \VERTEX\times\VERTEX \).

  The following set
  \begin{equation}
    \TopologyLower{\EDGE} = \bset{ \OpenSet \subset \VERTEX }%
    { \OpenSet\EDGE \subset \OpenSet }
    \label{eq:TopologyLowerEDGE}
  \end{equation}
  is an Alexandrov topology on~$\VERTEX$
  with the property that open subsets are characterized by 
  \begin{equation}
    \OpenSet \in \TopologyLower{\EDGE}
    \iff
    \OpenSet\EDGE \subset \OpenSet
    \iff      
    \OpenSet\TransitiveClosure{\EDGE}
    \subset \OpenSet
    \iff
    \OpenSet\TransitiveReflexiveClosure{\EDGE}\subset \OpenSet
    \iff
    \OpenSet\TransitiveReflexiveClosure{\EDGE}= \OpenSet
    \eqfinp
    \label{eq:TopologyLowerEDGE_OpenSet}
  \end{equation}
  In the Alexandrov topology~\( \TopologyLower{\EDGE} \),
  the topological closure\footnote{To alleviate the notation, we have
    denoted the topological closure by~\( \TopologicalClosure{\EDGE}{\Bgent} \)
    instead of~\( \TopologicalClosure{\TopologyLower{\EDGE}}{\Bgent} \).}~\( \TopologicalClosure{\EDGE}{\Bgent} \) 
  of a subset~\( \Bgent \subset \VERTEX\) is given by
  \begin{equation}
    \TopologicalClosure{\EDGE}{\Bgent} =
    \TransitiveReflexiveClosure{\EDGE}\Bgent
    \eqsepv \forall \Bgent  \subset \VERTEX
    \eqfinv 
    \label{eq:TopologicalClosureEDGE}
  \end{equation}
  that is, is the $\TransitiveReflexiveClosure{\EDGE}$-foreset;
the closed subsets are characterized by 
  \begin{equation}
    \ClosedSet \in \TopologyLower{\EDGE}'
    \iff
    \EDGE \ClosedSet \subset \ClosedSet
    \iff
    \TransitiveClosure{\EDGE}\ClosedSet \subset \ClosedSet
    \iff
    \TransitiveReflexiveClosure{\EDGE}\ClosedSet \subset \ClosedSet
    \iff
    \TransitiveReflexiveClosure{\EDGE}\ClosedSet = \ClosedSet
    \iff
    \TopologicalClosure{\EDGE}{\ClosedSet} = \ClosedSet
    \eqfinp 
    \label{eq:TopologyLowerEDGE_ClosedSet}
  \end{equation}
  The Alexandrov topology~\( \TopologyLower{\EDGE} \) satisfies
  \begin{subequations}
    \begin{align}
      \TopologyLower{\EDGE}
      &=
        \TopologyLower{\TransitiveClosure{\EDGE}}
        = \TopologyLower{\TransitiveReflexiveClosure{\EDGE}}
        \eqfinv 
        \intertext{and the dual Alexandrov topology~\( \TopologyLower{\EDGE}' \) satisfies}
        \TopologyLower{\EDGE}'
      &=
        \TopologyLower{\Converse{\EDGE}}
        = \TopologyLower{\npTransitiveClosure{\Converse{\EDGE}}}
        = \TopologyLower{\npTransitiveReflexiveClosure{\Converse{\EDGE}}}
        \eqfinp
    \end{align}
  \end{subequations}
  Regarding the specialization preorder~\( \moinsfine_\topology\), it is
  well-known that, for any preorder~$\moinsfine$ on~$\VERTEX$,
  we have that \( \moinsfine_{\TopologyLower{\moinsfine}}=\moinsfine \)
  \cite[Proposition~4.2.11]{Goubault-Larrecq:2013}.
  More generally, it holds that
  \begin{equation}
    \moinsfine_{\TopologyLower{\EDGE}} =
    \TransitiveReflexiveClosure{\EDGE}
    \eqsepv \forall \EDGE \subset \VERTEX^2
    \eqfinp 
  \end{equation}
\end{proposition}


\begin{proof}
  We prove that the set~\( \TopologyLower{\EDGE} \) in~\eqref{eq:TopologyLowerEDGE}
  contains both \( \emptyset, \VERTEX \) and is stable under 
  the union and intersection operations, be they finite or infinite, 
  which is what is required for an Alexandrov topology.
  Indeed, both \( \emptyset, \VERTEX \in \TopologyLower{\EDGE} \) as 
  \( \emptyset\EDGE=\emptyset\) and
  \( \VERTEX\EDGE \subset \VERTEX \).
  Let \( \sequence{\OpenSet_\scenario}{\scenario\in\SCENARIO} \) be a family in~\( \TopologyLower{\EDGE} \), that is,
  \( \OpenSet_\scenario\EDGE \subset \OpenSet_\scenario \) for all
  \( \scenario\in\SCENARIO \).
  We deduce that \( \np{\cup_{\scenario\in\SCENARIO}\OpenSet_\scenario}\EDGE
  = \cup_{\scenario\in\SCENARIO} \OpenSet_\scenario\EDGE
  \subset \cup_{\scenario\in\SCENARIO}\OpenSet_\scenario \), hence stability by
  union,
  and also that \( \np{\cap_{\scenario\in\SCENARIO}\OpenSet_\scenario}\EDGE
  \subset \cap_{\scenario\in\SCENARIO} \OpenSet_\scenario\EDGE
  \subset \cap_{\scenario\in\SCENARIO}\OpenSet_\scenario \), hence stability by intersection.

  We establish the useful equivalences:
  \begin{subequations}
    \begin{align*}
      \OpenSet\TransitiveReflexiveClosure{\EDGE}
      = \OpenSet
      & \iff
        \OpenSet\TransitiveReflexiveClosure{\EDGE}
        \subset \OpenSet
        \tag{because \( \OpenSet \subset
        \OpenSet\TransitiveReflexiveClosure{\EDGE} \) since
        \( \Delta \subset \TransitiveReflexiveClosure{\EDGE}=
        \TransitiveClosure{\EDGE}\cup\Delta \)}      
      \\
      & \iff      
        \OpenSet\TransitiveClosure{\EDGE}
        \subset \OpenSet
        \tag{because \( \TransitiveReflexiveClosure{\EDGE}=
        \TransitiveClosure{\EDGE}\cup\Delta \) }      
      \\
      & \iff
        \OpenSet\EDGE \subset \OpenSet
        \tag{because \( \TransitiveClosure{\EDGE} = \cup_{k=1}^{\infty}\EDGE^{k}
        \)
        and then by induction}      
      \\
      & \iff
        \EDGE \Complementary{\OpenSet} \subset \Complementary{\OpenSet}
        \intertext{indeed, suppose by contradiction that \( \OpenSet\EDGE \subset
        \OpenSet \) but that there exists \( \agent \in \EDGE \Complementary{\OpenSet}
        \) such that \( \agent \not\in  \Complementary{\OpenSet} \), that is,
        \( \agent \in  \OpenSet \);
        as a consequence, there exists \( \bgent \in \Complementary{\OpenSet} \)
        such that \( \agent \EDGE\bgent \), hence that \( \bgent \in \agent \EDGE\);
        now, as  \( \agent \in \OpenSet \), we get that \( \bgent \in  \agent
        \EDGE \subset \OpenSet\EDGE \subset\OpenSet\) by assumption; but this
        contradicts that \( \bgent \in \Complementary{\OpenSet} \); the reverse
        implication is proved in the same way ; the rest of the equivalences below
        are proved as above}
      & \iff
        \TransitiveClosure{\EDGE}\Complementary{\OpenSet} \subset \Complementary{\OpenSet}
      \\
      &     \iff
        \TransitiveReflexiveClosure{\EDGE}\Complementary{\OpenSet} \subset \Complementary{\OpenSet}
      \\
      & \iff
        \TransitiveReflexiveClosure{\EDGE}\Complementary{\OpenSet} = \Complementary{\OpenSet} 
        \eqfinp 
    \end{align*}
  \end{subequations}
  We deduce that~\eqref{eq:TopologyLowerEDGE_OpenSet} holds true,
  hence also that \( \TopologyLower{\EDGE}=
  \TopologyLower{\TransitiveClosure{\EDGE}}
  = \TopologyLower{\TransitiveReflexiveClosure{\EDGE}} \) 
  by~\eqref{eq:TopologyLowerEDGE},
  and that~\eqref{eq:TopologyLowerEDGE_ClosedSet} holds true,
  hence also that \( \TopologyLower{\EDGE}'= 
  \TopologyLower{\Converse{\EDGE}}
        = \TopologyLower{\npTransitiveClosure{\Converse{\EDGE}}}
        = \TopologyLower{\npTransitiveReflexiveClosure{\Converse{\EDGE}}} \) 
  by~\eqref{eq:TopologyLowerEDGE} and by~\eqref{eq:Topology'}.
  
  Finally, we consider a subset~\( \Bgent \subset \VERTEX\) and we characterize 
  its topological closure~\( \TopologicalClosure{\EDGE}{\Bgent} \), the smallest closed subset 
  that contains~\( \Bgent \). 
  On the one hand, we have that \( \Bgent \subset \TransitiveReflexiveClosure{\EDGE}\Bgent \) 
  since
  \( \TransitiveReflexiveClosure{\EDGE} =
  \TransitiveClosure{\EDGE}\cup\Delta \).
  On the other hand, the set \( \TransitiveReflexiveClosure{\EDGE}\Bgent \) 
  is closed since 
  \( \TransitiveReflexiveClosure{\EDGE}\np{\TransitiveReflexiveClosure{\EDGE}\Bgent}=
  \np{\TransitiveReflexiveClosure{\EDGE}}^2\Bgent
  = \TransitiveReflexiveClosure{\EDGE}\Bgent\), because the relation~\( \TransitiveReflexiveClosure{\EDGE} \) 
  is transitive.
  By definition of the topological closure~\( \TopologicalClosure{\EDGE}{\Bgent} \), 
  we deduce that \( \TopologicalClosure{\EDGE}{\Bgent} \subset \TransitiveReflexiveClosure{\EDGE}\Bgent \).
  Now, let \( \Cgent \subset \VERTEX \) be a closed subset such that 
  \( \Bgent \subset \Cgent \). We necessarily have that 
  \( \TransitiveReflexiveClosure{\EDGE}\Bgent \subset 
  \TransitiveReflexiveClosure{\EDGE}\Cgent =\Cgent\), where the last
  equality is 
by~\eqref{eq:TopologyLowerEDGE_ClosedSet} as $\Cgent$ is closed.
  
  As a consequence, the topological closure~\( \TopologicalClosure{\EDGE}{\Bgent}
  \) will always contain the closed set~\(
  \TransitiveReflexiveClosure{\EDGE}\Cgent \), from which we get that 
  \( \TransitiveReflexiveClosure{\EDGE}\Cgent \subset
  \TopologicalClosure{\EDGE}{\Bgent} \).
  We conclude that \( \TopologicalClosure{\EDGE}{\Bgent} = \TransitiveReflexiveClosure{\EDGE}\Bgent \).
  
  \medskip
  
  This ends the proof.
\end{proof}

\section{Equivalence between d-separation and t-separation}
\label{Equivalence_between_d-separation_and_t-separation}

In~\S\ref{d-_and_t-separation_between_vertices},
we recall the (extended) definition of d-separation,
then introduce a suitable topology on the set of vertices, 
and define a new notion of conditional topological separation (t-separation).
Then, we show that d-separation and t-separation between
vertices (and between subsets of vertices) are equivalent.
In~\S\ref{Characterization_of_t-separation_between_subsets},
we put forward a practical characterization of t-separation between subsets of
vertices.

\subsection{d- and t-separation between vertices}
\label{d-_and_t-separation_between_vertices}

We first recall the (extended) definition of d-separation, 
second define a new notion of conditional topological separation (t-separation)
and third prove their equivalence. 

\subsubsection{d-separation between vertices}
In the companion paper~\cite{Chancelier-De-Lara-Heymann-2021}
we generalize Pearl's d-separation beyond acyclic graphs as follows.

\begin{definition}(\cite[Definition 3]{Chancelier-De-Lara-Heymann-2021})
  \label{de:all_the_relations}
  Let $(\VERTEX,\EDGE)$ be a graph, that is,
$\VERTEX$ is a set and \( \EDGE \subset \VERTEX\times\VERTEX \),
  and let $\AgentSubsetW\subset\VERTEX$ be a subset of vertices.
  We define the \emph{conditional parental relation}~\( \ParentalPrecedence \)
  as
  \begin{subequations}
    \begin{align}
      \ParentalPrecedence 
      &= \Delta_{\Complementary{\AgentSubsetW}}\Precedence
        \mtext{ \qquad that is, }
        \bgent\ParentalPrecedence\cgent \iff
        \bgent\in\Complementary{\AgentSubsetW} \mtext{ and }
        \bgent\Precedence\cgent \qquad \bp{\forall \bgent,\cgent \in \AGENT }
        \eqfinv
        \label{eq:conditional_parental_relation}
        \intertext{the \emph{conditional ascendent relation}
        \( \ConditionalAscendent \) as }
        \ConditionalAscendent &=
                                \ConditionalDown = \Precedence \TransitiveReflexiveClosureParentalPrecedence
                                \mtext{ where } \TransitiveReflexiveClosureParentalPrecedence
                                = \npTransitiveReflexiveClosure{\ParentalPrecedence}
                                \label{eq:conditional_ascendent_relation}
                                \intertext{which relates a descendent with
                                an ascendent by means of 
                                elements in~$\Complementary{\AgentSubsetW}$. 
                                We define their converses~\( \ConverseParentalPrecedence \) and
                                \( \ConverseConditionalAscendent \) as }
                                \ConverseParentalPrecedence
      &= \npConverse{\ParentalPrecedence}
        = \Converse{\Precedence} \Delta_{\Complementary{\AgentSubsetW}}
        \eqfinv
        \label{eq:converse_conditional_parental_relation}
      \\
      \ConverseConditionalAscendent
      &= \Converse{\bp{\ConditionalAscendent}}
        = \ConditionalUp
        = \TransitiveReflexiveClosureConverseParentalPrecedence \Converse{\Precedence}
                                  \mtext{ where }  \TransitiveReflexiveClosureConverseParentalPrecedence
                                = \npTransitiveReflexiveClosure{\ConverseParentalPrecedence}
      \eqfinp
        \label{eq:converse_conditional_ascendent_relation}
        \intertext{With these elementary binary relations,
        we define the \emph{conditional common cause relation}~$\ConditionalCommonCause$ 
        as the symmetric relation}
        \ConditionalCommonCause 
                              &=
                                \ConverseConditionalAscendent \Delta_{\Complementary\AgentSubsetW}
                                \ConditionalAscendent
                                = \TransitiveClosureConverseParentalPrecedence \TransitiveClosureParentalPrecedence
                                \label{eq:common_cause}
                                \eqfinv
      \intertext{the \emph{conditional cousinhood relation}~$\Cousinhood$
      as the partial equivalence relation}
      \Cousinhood
      &=
        \bpTransitiveClosure{\Delta_{\AgentSubsetW} \ConditionalActiveTwo
        \Delta_{\AgentSubsetW} }
        \cup
        \Delta_{\AgentSubsetW}
        \eqfinv
        \label{eq:Cousinhood}
      \intertext{
      and the \emph{conditional active relation}~$\ConditionalActive$ 
      as the symmetric relation}
      \ConditionalActiveNew
      &= \Delta \cup 
        \ConditionalAscendent \cup \ConverseConditionalAscendent \cup \ConditionalActiveTwo
        \cup
        \bp{\ConditionalAscendent \cup \ConditionalActiveTwo}
        \Cousinhood
        \bp{\ConverseConditionalAscendent \cup \ConverseConditionalActiveTwo}
        \eqfinp
        \label{eq:conditional_active_relation}
    \end{align}
  \end{subequations}
\end{definition}

With the conditional active relation~$\ConditionalActive$, we can now define
the notion of d-separation between vertices
(which can readily be extended to d-separation between subsets of vertices).

\begin{definition}(d-separation between vertices,
  \cite[Definition 2]{Chancelier-De-Lara-Heymann-2021})
  \label{de:d-separated}
  Let $(\VERTEX,\EDGE)$ be a graph, that is,
$\VERTEX$ is a set and \( \EDGE \subset \VERTEX\times\VERTEX \),
  and let $\AgentSubsetW\subset\VERTEX$ be a subset of vertices.
  Let $\bgent$, $\cgent\in \AGENT$ be two vertices.
  We denote
  \begin{equation}
    \bgent \GenericSeparation{d} \cgent \mid \AgentSubsetW
    \iff 
    \neg \np{ \bgent\ConditionalActiveNew\cgent }
    \eqfinv
    \label{eq:d-separated}
  \end{equation}
  and we say that the vertices $\bgent$ and $\cgent$ are 
  \emph{d-separated} (\wrt\ $\AgentSubsetW$). 
\end{definition}

\subsubsection{t-separation between vertices}

We introduce a suitable topology on the set of vertices, 
and we define a new notion of conditional topological separation.

Let \( \npOrientedGraph \) be a graph, 
$\AgentSubsetW\subset\VERTEX$ be a subset of vertices,
and \( \ParentalPrecedence \) in~\eqref{eq:conditional_parental_relation}
be the corresponding conditional parental relation.
To alleviate the notation, in the Alexandrov topology
\( \TopologyLower{\ParentalPrecedence} \) in~\eqref{eq:TopologyLowerEDGE},
we use the following.
For any subset~\( \Bgent \subset \AGENT\),
the topological closure is denoted\footnote{
    Instead of~\( \TopologicalClosure{\ParentalPrecedence}{\Bgent} \)
or even of~\( \TopologicalClosure{\TopologyLower{\ParentalPrecedence}}{\Bgent} \).}
by~\( \TopologicalClosure{\AgentSubsetW}{\Bgent} \),
and the downset is denoted\footnote{%
  Instead of~\( \downarrow_{\TopologyLower{\ParentalPrecedence}} \).}
by~\( \downarrow_{\AgentSubsetW}\!\Bgent \).
By~\eqref{eq:TopologicalClosureAlexandrovTopology=LowerSet}
and~\eqref{eq:TopologicalClosureEDGE}, we get that   
\begin{equation}
  \TopologicalClosure{\AgentSubsetW}{\Bgent}
  = \ConditionalAncestor\Bgent
   = \, \downarrow_\AgentSubsetW \! \Bgent
  \eqfinp
  \label{eq:conditional_topology}
\end{equation}
Notice that the subset~\( \AgentSubsetW \) is \(
\TopologyLower{\ParentalPrecedence} \)-open, that is,
\( \AgentSubsetW \in \TopologyLower{\ParentalPrecedence} \).
Indeed, the complementary set~\( \Complementary{\AgentSubsetW} \) is closed
as it satisfies 
\( \ConditionalAncestor\Complementary{\AgentSubsetW} =
\npTransitiveClosure{\ParentalPrecedence}\Complementary{\AgentSubsetW}
\cup \Complementary{\AgentSubsetW} \subset \Complementary{\AgentSubsetW} \), as 
\( \ParentalPrecedence\, \AGENT \subset \Complementary{\AgentSubsetW} \)
because \( \ParentalPrecedence =
\Delta_{\Complementary{\AgentSubsetW}}\Precedence \) by~\eqref{eq:conditional_parental_relation}
and by definition of the subdiagonal relation~$\Delta_{\Complementary{\AgentSubsetW}}$.

With the conditional ascendent relation~\( \ConditionalAscendent \),
the conditional common cause relation~\( \ConverseConditionalActiveTwo \)
the conditional cousinhood relation~$\Cousinhood$
and the \( \TopologyLower{\ParentalPrecedence} \)-topological closure,   
we can now define the notion of t-separation between vertices
(which can readily be extended to t-separation between subsets of vertices).

\begin{definition}[Conditional topological separation between vertices, t-separation]
  \label{de:conditionally_topologically_separated}
  Let \( \npOrientedGraph \) be a graph, 
  and $\AgentSubsetW\subset\VERTEX$ be a subset of vertices.
  We set
  \begin{equation}
    \TopologicalRelation
    =\Delta \cup \Cousinhood
    \bp{\ConverseConditionalAscendent \cup \ConverseConditionalActiveTwo}
    \eqfinp
    \label{eq:TopologicalRelation}
  \end{equation}
  Let $\bgent$, $\cgent\in \AGENT$ be two vertices.
  We denote
  \begin{equation}
    \bgent \Tsep \cgent \mid \AgentSubsetW
    \iff 
    \TopologicalClosure{\AgentSubsetW}{\TopologicalRelation\bgent}
    \cap
    \TopologicalClosure{\AgentSubsetW}{\TopologicalRelation\cgent}
    = \emptyset
    \eqfinv
    \label{eq:conditionally_topologically_separated}
  \end{equation}
  and we say that the vertices $\bgent$ and $\cgent$ are 
  \emph{conditionally topologically separated} (\wrt\ $\AgentSubsetW$) or,
  shortly, \emph{t-separated}. 
\end{definition}

With the above definitions, we now show that the notions of d- and t-separation
are equivalent on the complementary set~$\Complementary{\AgentSubsetW}$.

\begin{theorem}
  \label{th:ConditionalDirectionalSeparation_IFF_topology}
  Let $(\VERTEX,\EDGE)$ be a graph, that is,
$\VERTEX$ is a set and \( \EDGE \subset \VERTEX\times\VERTEX \),
and let $\AgentSubsetW\subset\VERTEX$ be a subset of vertices.
%
   We have the equivalence
  \begin{equation}
    \bgent \GenericSeparation{t} \cgent \mid \AgentSubsetW
 \iff
    \bgent \GenericSeparation{d} \cgent \mid \AgentSubsetW
    \qquad \bp{ \forall \bgent,\cgent \in \Complementary{\AgentSubsetW} }
    \eqfinp
    \label{eq:ConditionalDirectionalSeparation_IFF_topology}      
  \end{equation}
\end{theorem}

\begin{proof}
  To prove~\eqref{eq:ConditionalDirectionalSeparation_IFF_topology},
  it is equivalent, by Definition~\ref{de:d-separated}, to prove the equivalence 
    \begin{equation}
    \bgent\Tsep\cgent
    \; \mid \AgentSubsetW  \iff
    \neg \np{ \bgent\ConditionalActiveNew\cgent }
    \qquad \bp{ \forall \bgent,\cgent \in \Complementary{\AgentSubsetW} }
    \eqfinp
    \label{eq:ConditionalDirectionalSeparation_IFF_topology_bis}      
  \end{equation}
For this purpose, we set
  \begin{equation}
    \ConverseTopologicalRelation
    =\Delta \cup
    \bp{\ConditionalAscendent \cup \ConditionalActiveTwo}\Cousinhood
    = \npConverse{\TopologicalRelation}
    \eqfinp 
    \label{eq:ConverseTopologicalRelation}
  \end{equation}
  Let \( \bgent, \cgent \in\AGENT \) be two vertices such that
  \( \bgent, \cgent \in \Complementary{\AgentSubsetW} \)
  . 
  We have 
  \begin{align*}
    \TopologicalClosure{\AgentSubsetW}{\TopologicalRelation\bgent}
    \cap
    \TopologicalClosure{\AgentSubsetW}{\TopologicalRelation\cgent}
    \neq \emptyset
    &  \iff
     \bp{ \ConditionalAncestor\TopologicalRelation\bgent }
      \cap
      \bp{ \ConditionalAncestor\TopologicalRelation\cgent }
      \neq \emptyset
      \intertext{because the topological closure of a subset~$\Bgent \subset \AGENT$ is given by
      \( \TopologicalClosure{\AgentSubsetW}{\Bgent}= \ConditionalAncestor\Bgent
      \) by~\eqref{eq:conditional_topology} }
    &  \iff
    \bp{ \bgent\ConverseTopologicalRelation\npConverse{\ConditionalAncestor} }
      \cap
   \bp{ \ConditionalAncestor\TopologicalRelation\cgent }
      \neq \emptyset
      \tag{by definition of the converse relation}
    \\
    &  \iff
      \bgent \ConverseTopologicalRelation \ConverseConditionalAncestor
      \ConditionalAncestor\TopologicalRelation\cgent
      \intertext{by definition of relation composition and by
      \( \npConverse{\ConditionalAncestor} = \ConverseConditionalAncestor \)}      
    &  \iff
      \bgent\Delta_{\Complementary{\AgentSubsetW}}\ConverseTopologicalRelation\ConverseConditionalAncestor
      \ConditionalAncestor\TopologicalRelation\Delta_{\Complementary{\AgentSubsetW}}\cgent
      \tag{ because  \( \bgent, \cgent \in \Complementary{\AgentSubsetW} \)
      by assumption
      }      
    \\
    &  \iff
      \bgent\Delta_{\Complementary{\AgentSubsetW}}
      \ConditionalActive
      \Delta_{\Complementary{\AgentSubsetW}}\cgent
      \tag{by~\eqref{eq:equation_for_topological_equivalence} in Appendix~\ref{Additional_Lemmas}}
    \\
    &  \iff
      \bgent\ConditionalActive\cgent
      \eqfinp       
      \tag{ because  \( \bgent, \cgent \in \Complementary{\AgentSubsetW} \)
      by assumption
      }      
  \end{align*}
  Thus, by taking the negation, we have
  obtained~\eqref{eq:ConditionalDirectionalSeparation_IFF_topology_bis},
  hence~\eqref{eq:ConditionalDirectionalSeparation_IFF_topology}
  by Definition~\ref{de:d-separated}. 
\end{proof}

\subsection{Characterization of t-separation between subsets}
\label{Characterization_of_t-separation_between_subsets}

We put forward a practical characterization of t-separation between subsets of
vertices. 
For this purpose, we introduce the notion of splitting, which slightly
generalizes the notion of partition. 

For any subset \( \Bgent \subset \AGENT\) and for any family
\( \sequence{\Bgent_\scenario}{\scenario\in\SCENARIO} \) of subsets
\( \Bgent_\scenario \subset \AGENT\),
we write \( \sqcup_{\scenario\in\SCENARIO}\Bgent_\scenario = \Bgent \)
when we have, on the one hand, 
\( \bp{ \scenario \neq \scenario' \implies \Bgent_\scenario \cap
  \Bgent_{\scenario'}  = \emptyset } \)
and, on the other hand, \( \bigcup_{\scenario\in\SCENARIO}\Bgent_\scenario = \Bgent \). 
We will also say that 
\( \sequence{\Bgent_\scenario}{\scenario\in\SCENARIO} \) 
is a \emph{splitting} of~\( \Bgent \) (we do not use the vocable of partition because it is not required that
the subsets~\( \Bgent_\scenario \) be nonempty).

\begin{proposition}[Topological separation between subsets]
  Let \( \npOrientedGraph \) be a graph, 
  and $\AgentSubsetW\subset\VERTEX$ be a subset of vertices.
  Let $\Bgent, \Cgent \subset \AGENT$ be two subsets of vertices such that 
  \begin{equation}
    \Bgent \cap \Cgent = \emptyset \eqsepv
    \Bgent \cap \AgentSubsetW=\emptyset \eqsepv
    \Cgent \cap \AgentSubsetW=\emptyset
    \eqfinp
    \label{eq:empty_intersection_Bgent_Cgent_AgentSubsetW}
  \end{equation}
  The following statements are equivalent:
  \begin{enumerate}
  \item
    \label{it:topologically_separated_one}
    For any \( \bgent\in\Bgent \), \( \cgent\in\Cgent \), we have that
    $\bgent \Tsep \cgent \mid \AgentSubsetW$, as in Definition~\ref{de:conditionally_topologically_separated},
  \item
    \label{it:topologically_separated_two}    
    There exists a splitting
    \( \AgentSubsetW_\Bgent , \AgentSubsetW_\Cgent \) of~$\AgentSubsetW$ 
    such that 
    \begin{equation}
      \AgentSubsetW_\Bgent \sqcup \AgentSubsetW_\Cgent = \AgentSubsetW 
      \text{ and } 
      \TopologicalClosure{\AgentSubsetW}{\Bgent \cup \AgentSubsetW_\Bgent}
      \cap 
      \TopologicalClosure{\AgentSubsetW}{\Cgent \cup \AgentSubsetW_\Cgent}
      = \emptyset 
      \eqfinp
      \label{eq:topologically_separated_TopologicalClosure}      
    \end{equation}
  \end{enumerate}
  \label{pr:topologically_separated}
\end{proposition}

\begin{proof}\quad 

  \noindent $\bullet$ (Item~\ref{it:topologically_separated_one} $\implies$ Item~\ref{it:topologically_separated_two}).
  We consider two subsets $\Bgent, \Cgent \subset \AGENT$ such that \eqref{eq:empty_intersection_Bgent_Cgent_AgentSubsetW}
holds true, and we prove the existence of a a splitting
  \( \AgentSubsetW_\Bgent , \AgentSubsetW_\Cgent \) of~$\AgentSubsetW$
  satisfying~\eqref{eq:empty_intersection_Bgent_Cgent_AgentSubsetW}
  in two steps. 
\medskip
  
  First, we set
  \( \AgentSubsetW'_\Bgent = \Cousinhood
  \bp{\ConverseConditionalAscendent \cup \ConverseConditionalActiveTwo}\Bgent
  \subset \AgentSubsetW \) by definition~\eqref{eq:Cousinhood} of~$\Cousinhood$ 
  and
  \( \AgentSubsetW'_\Cgent =\Cousinhood
  \bp{\ConverseConditionalAscendent \cup \ConverseConditionalActiveTwo}\Cgent
  \subset \AgentSubsetW \), and we prove that
  \begin{equation}
    \TopologicalClosure{\AgentSubsetW}{ \Bgent \cup \AgentSubsetW'_\Bgent }
    \cap
    \TopologicalClosure{\AgentSubsetW}{ \Cgent \cup \AgentSubsetW'_\Cgent }
    = \emptyset
   \label{eq:topologically_separated_proof_1}   
  \end{equation}
  We have that
  \begin{align*}
    \TopologicalClosure{\AgentSubsetW}{ \Bgent \cup \AgentSubsetW'_\Bgent }
    \cap
    \TopologicalClosure{\AgentSubsetW}{ \Cgent \cup \AgentSubsetW'_\Cgent }
    &=
      \TopologicalClosure{\AgentSubsetW}{ \Bgent \cup
      \Cousinhood
  \bp{\ConverseConditionalAscendent \cup \ConverseConditionalActiveTwo}\Bgent }
      \cap
      \TopologicalClosure{\AgentSubsetW}{ \Cgent \cup
      \bp{\ConverseConditionalAscendent \cup
      \ConverseConditionalActiveTwo}\Cgent }
     \\    
    &=
      \TopologicalClosure{\AgentSubsetW}{\TopologicalRelation\Bgent}
      \cap
      \TopologicalClosure{\AgentSubsetW}{\TopologicalRelation\Cgent}
      \tag{by definition~\eqref{eq:TopologicalRelation} of~$\TopologicalRelation$}
    \\    
    &= \bp{ \bigcup_{\bgent\in\Bgent} 
      \TopologicalClosure{\AgentSubsetW}{\TopologicalRelation\bgent} }
      \cap \bp{ \bigcup_{\cgent\in\Cgent} 
      \TopologicalClosure{\AgentSubsetW}{\TopologicalRelation\cgent} }
      \intertext{by property~\eqref{eq:TopologicalClosureAlexandrovTopology} of the
      topological closure in an Alexandrov topology}
    &=       
      \bigcup_{\bgent\in\Bgent,\cgent\in\Cgent} \bp{ 
      \underbrace{ \TopologicalClosure{\AgentSubsetW}{\TopologicalRelation\bgent}
      \cap \TopologicalClosure{\AgentSubsetW}{\TopologicalRelation\cgent} }_{=
      \emptyset} }
  \end{align*}
  by~\eqref{eq:conditionally_topologically_separated}
  as  $\bgent \Tsep \cgent \mid \AgentSubsetW$ by assumption,
  with \( \bgent\in\Bgent \subset \Complementary{\AgentSubsetW} \)
  and \( \cgent\in\Cgent \subset \Complementary{\AgentSubsetW} \)
  by~\eqref{eq:empty_intersection_Bgent_Cgent_AgentSubsetW}. 
  
Thus, we have proven~\eqref{eq:topologically_separated_proof_1}.
%
\medskip

  Second, we are now going to prove that we can enlarge the subsets~$\AgentSubsetW'_\Bgent$ and
  $\AgentSubsetW'_\Cgent$ to obtain a splitting of~$\AgentSubsetW$ satisfying
  Equation~\eqref{eq:topologically_separated_TopologicalClosure}.
  
  For this purpose, we set 
  $\widetilde{\AgentSubsetW} = \AgentSubsetW \backslash \np{\AgentSubsetW'_\Bgent
    \cup \AgentSubsetW'_\Cgent} \subset \AgentSubsetW $.
  The cousinhood relation~$\Cousinhood$ in~\eqref{eq:Cousinhood} is a 
  partial equivalence relation on~$\VERTEX$, and it is easily seen to be an equivalence relation on~$\AgentSubsetW$.
  This is why we consider the partition $\widetilde{\AgentSubsetW}= \sqcup_{i\in
    I} \widetilde{\AgentSubsetW}_i $ of $\widetilde{\AgentSubsetW}$,
  where, for each $i\in I$, the elements of 
  the subset $\widetilde{\AgentSubsetW}_{i}$ belong to the same equivalence class
  of the equivalence relation \( \Cousinhood \subset \AgentSubsetW^2 \)
  (understood as the restriction of the cousinhood relation~$\Cousinhood$ to~$\AgentSubsetW$).
  We are now going to prove that
  \begin{equation}
    \label{eq:decomposeWtilde}
\forall i\in I     \eqsepv 
    \text{  either  }\quad
    \TopologicalClosure{\AgentSubsetW}{\widetilde{\AgentSubsetW}_i} \cap \TopologicalClosure{\AgentSubsetW}{ \Bgent \cup \AgentSubsetW'_\Bgent } =
    \emptyset \quad \text{  or  }\quad
    \TopologicalClosure{\AgentSubsetW}{\widetilde{\AgentSubsetW}_i} \cap \TopologicalClosure{\AgentSubsetW}{ \Cgent \cup \AgentSubsetW'_\Cgent } =
    \emptyset
    \eqfinp
  \end{equation}
  The proof is by contradiction. 
  Let $i\in I$ be fixed and suppose that both
  \( \TopologicalClosure{\AgentSubsetW}{\widetilde{\AgentSubsetW}_i} \cap
  \TopologicalClosure{\AgentSubsetW}{ \Bgent \cup \AgentSubsetW'_\Bgent }
  \not= \emptyset \)  and  \(
  \TopologicalClosure{\AgentSubsetW}{\widetilde{\AgentSubsetW}_i} \cap
  \TopologicalClosure{\AgentSubsetW}{ \Cgent \cup \AgentSubsetW'_\Cgent }
  \not= \emptyset \).
  First, as $\TopologicalClosure{\AgentSubsetW}{\widetilde{\AgentSubsetW}_i} \cap
  \TopologicalClosure{\AgentSubsetW}{ \Bgent \cup \AgentSubsetW'_\Bgent } \not=\emptyset$,
  there would exist \( \cgent \in \TopologicalClosure{\AgentSubsetW}{\widetilde{\AgentSubsetW}_i} \cap
  \TopologicalClosure{\AgentSubsetW}{ \Bgent \cup \AgentSubsetW'_\Bgent } \),
  hence there would exist $w_{i}^{(1)} \in \widetilde{\AgentSubsetW}_i$ such that
  $\cgent \ConditionalAncestor w_{i}^{(1)}$ (as
  \( \TopologicalClosure{\AgentSubsetW}{\widetilde{\AgentSubsetW}_i}=
  \ConditionalAncestor\widetilde{\AgentSubsetW}_i \)
  by~\eqref{eq:TopologicalClosureEDGE}), 
  and there would exist   $\bgent \in \Bgent$ such that
  $\cgent\ConditionalAncestor
  \bp{ \Delta \cup \Cousinhood\np{\ConverseConditionalAscendent \cup \ConverseConditionalActiveTwo}} \bgent$ 
  (as \( \TopologicalClosure{\AgentSubsetW}{ \Bgent \cup \AgentSubsetW'_\Bgent }
  =\ConditionalAncestor\np{ \Bgent \cup \AgentSubsetW'_\Bgent }
  =  \ConditionalAncestor 
  \bp{ \Delta \cup \Cousinhood\np{\ConverseConditionalAscendent \cup \ConverseConditionalActiveTwo}} \Bgent
  \)
  by~\eqref{eq:TopologicalClosureEDGE} and by definition of~$\AgentSubsetW'_\Bgent$).
  Thus, we would have that \\
  \(     w_{i}^{(1)} \ConverseConditionalAncestor\ConditionalAncestor 
  \bp{ \Delta \cup \Cousinhood\np{\ConverseConditionalAscendent \cup \ConverseConditionalActiveTwo}}
  \bgent \).
  Second, as $\TopologicalClosure{\AgentSubsetW}{\widetilde{\AgentSubsetW}_i} \cap
  \TopologicalClosure{\AgentSubsetW}{ \Cgent \cup \AgentSubsetW'_\Cgent }
  \not=\emptyset$ and proceeding in the same way,
  there would exist $w_{i}^{(2)} \in \widetilde{\AgentSubsetW}_i$ and $\cgent \in \Cgent$ such that
  $w_{i}^{(2)} \ConverseConditionalAncestor\ConditionalAncestor 
  \bp{ \Delta \cup\Cousinhood \np{\ConverseConditionalAscendent \cup \ConverseConditionalActiveTwo}}
  \cgent$.
  Now, as $w_{i}^{(1)}$ and $w_{i}^{(2)}$ would be both in $\widetilde{\AgentSubsetW}_i$ 
  they would be in the same equivalence class for the equivalence
  relation~$\Cousinhood$, giving thus $w_{i}^{(1)}\Cousinhood
  w_{i}^{(2)}$. Finally, we would obtain that
  \begin{equation}
    \bgent\Delta_{\Complementary{\AgentSubsetW}} 
    \bp{ \Delta \cup \np{\ConditionalAscendent \cup \ConverseConditionalActiveTwo}\Cousinhood}
    \ConverseConditionalAncestor\ConditionalAncestor \Cousinhood \ConverseConditionalAncestor\ConditionalAncestor 
    \bp{ \Delta \cup\Cousinhood \np{\ConverseConditionalAscendent \cup \ConverseConditionalActiveTwo}}
    \Delta_{\Complementary{\AgentSubsetW}}\cgent
    \eqfinp
    \label{eq:gamma-rel-lambda}
  \end{equation}
  Combining~\eqref{eq:gamma-rel-lambda}, Equation~\eqref{eq:gamma-rel-lambda-1}
  and the definition~\eqref{eq:conditional_active_relation}
  of~$\ConditionalActiveNew$, we would obtain that 
  $\bgent\ConditionalActiveNew \cgent$.
  Thus, we would arrive at a contradiction as we have,
  by assumption, $\bgent \Tsep \cgent \mid \AgentSubsetW$, hence 
\( \neg \np{ \bgent\ConditionalActiveNew \cgent } \)
by~\eqref{eq:ConditionalDirectionalSeparation_IFF_topology_bis}. 

Thus, we have proven that the disjunction~\eqref{eq:decomposeWtilde} holds true,
from which we obtain a splitting $I = I_{\Bgent}\sqcup I_{\Cgent}$
and a splitting 
  $\widetilde{\AgentSubsetW}= \widetilde{\AgentSubsetW}_{\Bgent} \sqcup  \widetilde{\AgentSubsetW}_{\Cgent}$
  defined by 
  \begin{align}
    I_{\Bgent} = \bset{ i \in
    I}{\TopologicalClosure{\AgentSubsetW}{\widetilde{\AgentSubsetW}_i}
    \cap \TopologicalClosure{\AgentSubsetW}{ \Cgent \cup \AgentSubsetW'_\Cgent }
    = \emptyset}
    & \quad \text{ and  }\quad
      I_{\Cgent} = I \backslash I_{\Bgent} \eqfinv
      \nonumber  \\
    \widetilde{\AgentSubsetW}_{\Bgent}=
    \bigcup_{i\in I_{\Bgent}}
    \TopologicalClosure{\AgentSubsetW}{\widetilde{\AgentSubsetW}_i}
    & \quad \text{ and  }\quad
      \widetilde{\AgentSubsetW}_{\Cgent}=
      \bigcup_{i\in I_{\Cgent}}
      \TopologicalClosure{\AgentSubsetW}{\widetilde{\AgentSubsetW}_i}
      \eqfinv
      \nonumber
      \intertext{which, by construction, satisfies}
      \TopologicalClosure{\AgentSubsetW}{\widetilde{\AgentSubsetW}_{\Bgent}}
      \cap \TopologicalClosure{\AgentSubsetW}{ \Cgent \cup \AgentSubsetW'_\Cgent } =\emptyset
    & \quad\text{and}\quad
      \TopologicalClosure{\AgentSubsetW}{\widetilde{\AgentSubsetW}_{\Cgent}}
      \cap \TopologicalClosure{\AgentSubsetW}{ \Bgent \cup \AgentSubsetW'_\Bgent }= \emptyset
      \eqfinp 
      \label{eq:Wtildedef}     
  \end{align}
  
  Now, we define ${\AgentSubsetW}_{\Bgent}= {\AgentSubsetW}'_{\Bgent}\cup \widetilde{\AgentSubsetW}_{\Bgent}$
  and  ${\AgentSubsetW}_{\Cgent}= {\AgentSubsetW}'_{\Cgent}\cup
  \widetilde{\AgentSubsetW}_{\Cgent}$,
  and we are going to prove
  that~\eqref{eq:topologically_separated_TopologicalClosure} holds true.

  To check the second part
  of~\eqref{eq:topologically_separated_TopologicalClosure},
  we calculate  
  \begin{align*}
    \TopologicalClosure{\AgentSubsetW}{ \Bgent \cup \AgentSubsetW_\Bgent }
    \cap
    \TopologicalClosure{\AgentSubsetW}{ \Cgent \cup \AgentSubsetW_\Cgent }
    &=
      \TopologicalClosure{\AgentSubsetW}{ \Bgent \cup \AgentSubsetW'_\Bgent \cup \widetilde{\AgentSubsetW}_{\Bgent} }
      \cap
      \TopologicalClosure{\AgentSubsetW}{ \Cgent \cup \AgentSubsetW'_\Cgent\cup \widetilde{\AgentSubsetW}_{\Cgent} }
    \\
    &=
      \bp{\TopologicalClosure{\AgentSubsetW}{ \Bgent \cup \AgentSubsetW'_\Bgent} \cup
      \TopologicalClosure{\AgentSubsetW}{\widetilde{\AgentSubsetW}_{\Bgent} }}
      \cap
      \bp{
      \TopologicalClosure{\AgentSubsetW}{ \Cgent \cup \AgentSubsetW'_\Cgent}
      \cup
      \TopologicalClosure{\AgentSubsetW}{\widetilde{\AgentSubsetW}_{\Cgent} }}
      \tag{by~\eqref{eq:TopologicalClosureAlexandrovTopology}}
    \\
    &=
      \underbrace{\bp{\TopologicalClosure{\AgentSubsetW}{ \Bgent \cup \AgentSubsetW'_\Bgent}
      \cap
      \TopologicalClosure{\AgentSubsetW}{ \Cgent \cup \AgentSubsetW'_\Cgent}}
      }_{=\emptyset \text{ by~\eqref{eq:topologically_separated_proof_1} }}
      \cup
      \underbrace{ \bp{\TopologicalClosure{\AgentSubsetW}{ \Bgent \cup \AgentSubsetW'_\Bgent}
      \cap
      \TopologicalClosure{\AgentSubsetW}{\widetilde{\AgentSubsetW}_{\Cgent} }}
      }_{=\emptyset \text{ by~\eqref{eq:Wtildedef}}}
    \\
    &\hspace{1cm}
      \cup
      \underbrace{
      \bp{\TopologicalClosure{\AgentSubsetW}{\widetilde{\AgentSubsetW}_{\Bgent} }
      \cap
      \TopologicalClosure{\AgentSubsetW}{ \Cgent \cup \AgentSubsetW'_\Cgent}}
      }_{=\emptyset \text{ by~\eqref{eq:Wtildedef}}}
      \cup
      {\bp{\TopologicalClosure{\AgentSubsetW}{\widetilde{\AgentSubsetW}_{\Bgent} }
      \cap
      \TopologicalClosure{\AgentSubsetW}{\widetilde{\AgentSubsetW}_{\Cgent} }}
      }
    \\
    &=
      \bp{\bigcup_{i\in I_{\Bgent}} \TopologicalClosure{\AgentSubsetW}{\widetilde{\AgentSubsetW}_i}}
      \cap
      \bp{ 
      \bigcup_{j\in I_{\Cgent}}
      \TopologicalClosure{\AgentSubsetW}{\widetilde{\AgentSubsetW}_j}}
      \intertext{by~\eqref{eq:TopologicalClosureAlexandrovTopology} and by definition
      of
      \( \widetilde{\AgentSubsetW}_{\Bgent} \) and \( \widetilde{\AgentSubsetW}_{\Cgent} \)}
    &=
      \bigcup_{i\in I_{\Bgent}} \bigcup_{j\in I_{\Cgent}}
      \underbrace{ \bp{
      \TopologicalClosure{\AgentSubsetW}{\widetilde{\AgentSubsetW}_i} 
      \cap
      \TopologicalClosure{\AgentSubsetW}{\widetilde{\AgentSubsetW}_j} }
      }_{=\emptyset}
      = \emptyset
  \end{align*}
  as \( I_{\Bgent} \cap I_{\Cgent} =  \emptyset \), 
  using the postponed Lemma~\ref{lem:cousinhood-and-closure}
  which gives \( \TopologicalClosure{\AgentSubsetW}{\widetilde{\AgentSubsetW}_i} 
    \cap 
    \TopologicalClosure{\AgentSubsetW}{\widetilde{\AgentSubsetW}_j} =\emptyset
    \), for any \( i,j \in I \) with \( i\not=j \).
    From the just proven equality \( \TopologicalClosure{\AgentSubsetW}{ \Bgent \cup \AgentSubsetW_\Bgent }
    \cap
    \TopologicalClosure{\AgentSubsetW}{ \Cgent \cup \AgentSubsetW_\Cgent }
    =\emptyset \), we readily get that
    \( {\AgentSubsetW}_{\Bgent} \cap {\AgentSubsetW}_{\Cgent}   =\emptyset \).
      Therefore, to check the first part
  of~\eqref{eq:topologically_separated_TopologicalClosure}, it remains 
  to calculate
  \[
    {\AgentSubsetW}_{\Bgent} \cup {\AgentSubsetW}_{\Cgent}
=
      \np{\AgentSubsetW'_\Bgent\cup \AgentSubsetW'_\Cgent}\cup
\underbrace{ \bigcup_{i\in I_{\Bgent}}
    \TopologicalClosure{\AgentSubsetW}{\widetilde{\AgentSubsetW}_i}
\cup
      \bigcup_{i\in I_{\Cgent}}
      \TopologicalClosure{\AgentSubsetW}{\widetilde{\AgentSubsetW}_i} }_{ =\widetilde{\AgentSubsetW}= \AgentSubsetW \backslash
      \np{\AgentSubsetW'_\Bgent\cup \AgentSubsetW'_\Cgent} }
=      
      \np{\AgentSubsetW'_\Bgent\cup \AgentSubsetW'_\Cgent}\cup
     \bp{ \AgentSubsetW \backslash
      \np{\AgentSubsetW'_\Bgent\cup \AgentSubsetW'_\Cgent} }
      = \AgentSubsetW
      \eqfinp
 \]
  \medskip

  \noindent $\bullet$ (Item~\ref{it:topologically_separated_two} $\implies$ Item~\ref{it:topologically_separated_one}).
The proof is by contradiction.
For this purpose,
we suppose     we suppose that there exists a splitting
    $\AgentSubsetW=\AgentSubsetW_\Bgent \sqcup \AgentSubsetW_\Cgent $ such that
  $\TopologicalClosure{\AgentSubsetW}{\Bgent \cup \AgentSubsetW_\Bgent}
  \cap 
  \TopologicalClosure{\AgentSubsetW}{\Cgent \cup \AgentSubsetW_\Cgent}
  = \emptyset$ (Item~\ref{it:topologically_separated_two})
  and that there exists $\bgent \in \Bgent$ and $\cgent \in \Cgent$ such that
    $\bgent\ConditionalActive\cgent$
    ($\neg$~Item~\ref{it:topologically_separated_one}). 
  We show that we arrive at a contradiction.

  Using the fact that $\bgent\ConditionalActive\cgent$
  and that $
  \TopologicalClosure{\AgentSubsetW}{\na{\bgent}}
  \cap 
  \TopologicalClosure{\AgentSubsetW}{\na{\cgent}}
  \subset
  \TopologicalClosure{\AgentSubsetW}{\Bgent \cup \AgentSubsetW_\Bgent}
  \cap 
  \TopologicalClosure{\AgentSubsetW}{\Cgent \cup \AgentSubsetW_\Cgent}
  = \emptyset$, 
  we would obtain, by the postponed Lemma~\ref{lem:TopologicalAw}, that there
  would exist a nonempty subset $\AgentSubsetW_{\bgent,\cgent} \subset \AgentSubsetW$ such that all the elements of 
  $\AgentSubsetW_{\bgent,\cgent}$ would be in the same class for
  the~$\Cousinhood$ partial equivalence relation, 
  and such that 
  \begin{equation}
    \TopologicalClosure{\AgentSubsetW}{\na{\bgent}} 
    \cap \TopologicalClosure{\AgentSubsetW}{\AgentSubsetW_{\bgent,\cgent}} \not=\emptyset 
    \quad\text{and}\quad
    \TopologicalClosure{\AgentSubsetW}{\na{\cgent}} 
    \cap \TopologicalClosure{\AgentSubsetW}{\AgentSubsetW_{\bgent,\cgent}} \not=\emptyset
    \eqfinp
    \label{eq:Wgammalambdaintersect}
  \end{equation}
  Now, using the fact that
  $\TopologicalClosure{\AgentSubsetW}{\AgentSubsetW_\Bgent}
  \cap 
  \TopologicalClosure{\AgentSubsetW}{\AgentSubsetW_\Cgent}
  \subset 
  \TopologicalClosure{\AgentSubsetW}{\Bgent \cup \AgentSubsetW_\Bgent}
  \cap 
  \TopologicalClosure{\AgentSubsetW}{\Cgent \cup \AgentSubsetW_\Cgent}
  =\emptyset $, we would obtain that $\TopologicalClosure{\AgentSubsetW}{\AgentSubsetW_\Bgent}
  \cap 
  \TopologicalClosure{\AgentSubsetW}{\AgentSubsetW_\Cgent}=\emptyset$,
  which, would imply that
$\AgentSubsetW_{\bgent,\cgent}$ would necessarily be included either in $\AgentSubsetW_\Bgent$ 
or in $\AgentSubsetW_\Cgent$, by using the second part of
  Lemma~\ref{lem:cousinhood-and-closure}.
  Therefore, 
  Assuming that $\AgentSubsetW_{\bgent,\cgent} \subset \AgentSubsetW_\Bgent$, we
  obtain a contradiction
  using Equation~\eqref{eq:Wgammalambdaintersect}, as we would have that
  \[
    \emptyset \not= \TopologicalClosure{\AgentSubsetW}{\na{\cgent}}
    \cap  \TopologicalClosure{\AgentSubsetW}{\AgentSubsetW_{\bgent,\cgent}}
    \subset  
    \TopologicalClosure{\AgentSubsetW}{\na{\cgent}} \cap  \TopologicalClosure{\AgentSubsetW}{\AgentSubsetW_\Bgent}
    \subset 
    \TopologicalClosure{\AgentSubsetW}{\Cgent \cup \AgentSubsetW_\Cgent}
    \cap 
    \TopologicalClosure{\AgentSubsetW}{\Bgent \cup \AgentSubsetW_\Bgent}
    =\emptyset
    \eqfinp
  \]
  Proceeding in a similar way in the case $\AgentSubsetW_{\bgent,\cgent} \subset \AgentSubsetW_\Cgent$ 
  leads to a similar contradiction.
\end{proof}



We end with Lemma~\ref{lem:TopologicalAw} and
Lemma~\ref{lem:cousinhood-and-closure} which are instrumental in the proof of
Proposition~\ref{pr:topologically_separated}.

\begin{lemma}
  Suppose that the assumptions of Proposition~\ref{pr:topologically_separated}
  are satisfied. 
  Let \( \bgent\in\Bgent \) and \( \cgent\in\Cgent \) be given such that
  $\bgent\ConditionalActive\cgent$
  and $
  \TopologicalClosure{\AgentSubsetW}{\na{\bgent}}
  \cap 
  \TopologicalClosure{\AgentSubsetW}{\na{\cgent}} = \emptyset$.
  Then, there exists $w_{\bgent}$, $w_{\cgent}\in \AgentSubsetW$ such that
  $w_{\bgent}$ and  $w_{\cgent}$ are in the same equivalence class of the
  partial equivalence relation $\Cousinhood$ (that is, $w_{\bgent} \Cousinhood
  w_{\cgent}$) and
  such that $\TopologicalClosure{\AgentSubsetW}{\na{\bgent}}\cap\TopologicalClosure{\AgentSubsetW}{\na{w_\bgent}}\not=\emptyset$ and
  $\TopologicalClosure{\AgentSubsetW}{\na{\cgent}}\cap\TopologicalClosure{\AgentSubsetW}{\na{w_\cgent}}\not=\emptyset$.
  \label{lem:TopologicalAw}
\end{lemma}

\begin{proof}
  As a preliminary result, 
  we prove that $\bgent \bp{\Delta \cup \ConditionalAscendent \cup \ConverseConditionalAscendent \cup \ConditionalActiveTwo}\cgent$
  contradicts the assumptions of Lemma~\ref{lem:TopologicalAw}.
  First, $\bgent \Delta \cgent$ contradicts the assumption $\Bgent \cap \Cgent =
  \emptyset$ in~\eqref{eq:empty_intersection_Bgent_Cgent_AgentSubsetW}.
  Second, $\bgent \ConditionalAscendent \cgent$ contradicts the assumption
  $ \TopologicalClosure{\AgentSubsetW}{\na{\bgent}} \cap \TopologicalClosure{\AgentSubsetW}{\na{\cgent}}=\emptyset$.
  Indeed, using the definition~\eqref{eq:conditional_ascendent_relation}
  of the conditional ascendent relation~$\ConditionalAscendent$, and using the fact that
  $\bgent \in \Bgent \subset \complementary{\AgentSubsetW}$, we have that
  $\bgent \ConditionalAscendent \cgent =\bgent \Precedence \TransitiveReflexiveClosureParentalPrecedence  \cgent
  = \bgent \Delta_{\complementary{\AgentSubsetW}} \Precedence \TransitiveReflexiveClosureParentalPrecedence  \cgent =
  \bgent \TransitiveReflexiveClosureParentalPrecedence  \cgent$. Thus, $\bgent \ConditionalAscendent \cgent$ implies that
  $\bgent \in \TransitiveReflexiveClosureParentalPrecedence  \cgent=
  \TopologicalClosure{\AgentSubsetW}{\na{\cgent}}$ and thus a contradiction as
  \( \emptyset = \TopologicalClosure{\AgentSubsetW}{\na{\bgent}} \cap
  \TopologicalClosure{\AgentSubsetW}{\na{\cgent}}
  \supset \bgent \cap \TopologicalClosure{\AgentSubsetW}{\na{\cgent}} \).
  Third, following the same lines, $\bgent \ConverseConditionalAscendent\cgent$ implies that
  $\cgent \in \TopologicalClosure{\AgentSubsetW}{\na{\bgent}}$ and contradicts $
  \TopologicalClosure{\AgentSubsetW}{\na{\bgent}} \cap
  \TopologicalClosure{\AgentSubsetW}{\na{\cgent}}=\emptyset$.
  Fourth, $\bgent \ConditionalActiveTwo \cgent= \bgent \TransitiveClosureConverseParentalPrecedence \TransitiveClosureParentalPrecedence\cgent$
  and thus $\bgent \ConditionalActiveTwo \cgent$ implies that
  \( \TransitiveClosureParentalPrecedence\bgent \cap
  \TransitiveClosureParentalPrecedence\cgent \neq \emptyset \),
  hence that \( \TransitiveReflexiveClosureParentalPrecedence \bgent \cap
  \TransitiveReflexiveClosureParentalPrecedence\cgent \neq \emptyset \),
  that is, 
  $\TopologicalClosure{\AgentSubsetW}{\na{\bgent}} \cap
  \TopologicalClosure{\AgentSubsetW}{\na{\cgent}}\not=\emptyset$
  by~\eqref{eq:conditional_topology}.
  This again leads to a contradiction.

  Therefore, we get that 
  \( \neg \bp{ \bgent \bp{\Delta \cup \ConditionalAscendent \cup
      \ConverseConditionalAscendent \cup \ConditionalActiveTwo}\cgent } \) and 
  $\bgent\ConditionalActive\cgent$ imply that  
  $\bgent \bp{\ConditionalAscendent \cup \ConditionalActiveTwo}
  \Cousinhood
  \bp{\ConverseConditionalAscendent \cup \ConverseConditionalActiveTwo} \cgent$,
  using the definition~\eqref{eq:conditional_active_relation}
  of the conditional active relation~$\ConditionalActive$.
  Thus, there exist
  $w_{\bgent}$, $w_{\cgent}\in \AgentSubsetW$ such that 
  \begin{equation*}
    \bgent \bp{\ConditionalAscendent \cup \ConditionalActiveTwo}w_{\bgent}
    \text{ and }
    w_{\cgent} \bp{\ConverseConditionalAscendent \cup \ConverseConditionalActiveTwo} \cgent 
    \text{ and }
    w_{\bgent} \Cousinhood w_{\cgent}
    \eqfinp
  \end{equation*}
  
  Now, we prove that
  $\bgent \bp{\ConditionalAscendent \cup \ConditionalActiveTwo}w_{\bgent}$
  implies that we have
  $\TopologicalClosure{\AgentSubsetW}{\na{\bgent}}\cap\TopologicalClosure{\AgentSubsetW}{\na{w_\bgent}}\not=\emptyset$. Indeed,
  as $\bgent \bp{\ConditionalAscendent \cup \ConditionalActiveTwo}w_{\bgent}$,
  we have two possibilities.
  First, suppose that $\bgent \ConditionalAscendent w_{\bgent}$. Then, as
  $\bgent \in \Complementary{\AgentSubsetW}$, this implies that
  $\bgent \Delta_{\Complementary{\AgentSubsetW}} \ConditionalAscendent
  w_{\bgent} = \bgent \TransitiveReflexiveClosureParentalPrecedence w_{\bgent}$
  which implies that
  $\bgent \in \TransitiveReflexiveClosureParentalPrecedence w_{\bgent}=
  \TopologicalClosure{\AgentSubsetW}{\na{w_\bgent}}$
  by~\eqref{eq:conditional_topology}.  Therefore, we get that
  $\TopologicalClosure{\AgentSubsetW}{\na{\bgent}}\cap\TopologicalClosure{\AgentSubsetW}{\na{w_\bgent}}\not=\emptyset$.
  Second, if $\bgent \ConditionalActiveTwo w_{\bgent}$, then, as already seen at the beginning of the proof, we obtain that 
  $\TopologicalClosure{\AgentSubsetW}{\na{\bgent}}\cap\TopologicalClosure{\AgentSubsetW}{\na{w_\bgent}}\not=\emptyset$. 
  
  Then, following similar arguments, 
  we prove that $w_{\cgent} \bp{\ConverseConditionalAscendent \cup \ConverseConditionalActiveTwo} \cgent$ implies that 
  we have $\TopologicalClosure{\AgentSubsetW}{\na{\cgent}}\cap\TopologicalClosure{\AgentSubsetW}{\na{w_\cgent}}\not=\emptyset$.
  
  Finally, we have obtained that 
  $\TopologicalClosure{\AgentSubsetW}{\na{\bgent}}\cap\TopologicalClosure{\AgentSubsetW}{\na{w_\bgent}}\not=\emptyset$ and 
  $\TopologicalClosure{\AgentSubsetW}{\na{\cgent}}\cap\TopologicalClosure{\AgentSubsetW}{\na{w_\cgent}}\not=\emptyset$. 
  Moreover $w_{\bgent}$ and $w_{\cgent}$ are in the same equivalence class of the partial equivalence
  relation $\Cousinhood$ as $w_{\bgent} \Cousinhood w_{\cgent}$.
  Thus, we have found two elements $w_{\bgent}$, $w_{\cgent} \in \AgentSubsetW$ satisfying the conclusion of 
  Lemma~\ref{lem:TopologicalAw}.
  This concludes the proof.
\end{proof}

\begin{lemma} Let $\AgentSubsetW'$ and $\AgentSubsetW''$ be two subsets of $\AgentSubsetW$ which are included in
  two distinct equivalence classes of the partial equivalence
  relation~$\Cousinhood$. 
  Then, we have that $\TopologicalClosure{\AgentSubsetW}{\AgentSubsetW'}\cap\TopologicalClosure{\AgentSubsetW}{\AgentSubsetW''}
  =\emptyset$.
  
Conversely, assume given a splitting
    $\AgentSubsetW=\AgentSubsetW'\sqcup \AgentSubsetW''$ such that
  $\TopologicalClosure{\AgentSubsetW}{\AgentSubsetW'}\cap\TopologicalClosure{\AgentSubsetW}{\AgentSubsetW''}=\emptyset$.
  Then, there does not exists $w' \in \AgentSubsetW'$ and 
  $w'' \in \AgentSubsetW''$ such that $w'$ and $w''$ are in the same equivalence classes of~$\Cousinhood$
  
  \label{lem:cousinhood-and-closure}
\end{lemma}

\begin{proof} For the first assertion, we make a proof by contradiction.
  For this purpose, we consider $\AgentSubsetW'$ and $\AgentSubsetW''$, two subsets of $\AgentSubsetW$ which are included in
  two distinct equivalence classes of the partial equivalence
  relation~$\Cousinhood$, and we 
suppose that $\TopologicalClosure{\AgentSubsetW}{\AgentSubsetW'}\cap\TopologicalClosure{\AgentSubsetW}{\AgentSubsetW''}
  \not=\emptyset$. Then, by~\eqref{eq:TopologicalClosureEDGE}, 
  there exists $w'\in \AgentSubsetW'$ and $w''\in \AgentSubsetW''$ and $\bgent \in \VERTEX$ such that 
  $\bgent \in \ConditionalAncestor w'$ and $\bgent \in  \ConditionalAncestor w''$. Therefore we have that
  $w' \ConverseConditionalAncestor \ConditionalAncestor w''$, from which we
  deduce that
  $w'  \np{\Delta 
    \cup
    \Delta_{\Complementary{\AgentSubsetW}} \ConditionalAscendent 
    \cup \ConverseConditionalAscendent \Delta_{\Complementary{\AgentSubsetW}}
    \cup \ConditionalActiveTwo
  } w''$,  using Equation~\eqref{eq:AA}.
  Moreover, as by assumption, $\AgentSubsetW'$ and $\AgentSubsetW''$ are two subsets of $\AgentSubsetW$ which are included in two distinct equivalence classes of $\Cousinhood$, we have that $w'\in \AgentSubsetW$, $w''\in \AgentSubsetW$ and
  $w' \not= w''$. Hence, among the four possible cases corresponding to the
  union of four terms, only the last one is possible:
  we must necessarily have that $w' \Delta_{\AgentSubsetW} \ConditionalActiveTwo \Delta_{\AgentSubsetW}  w''$,
  which implies that $w' \Cousinhood  w''$ by definition~\eqref{eq:Cousinhood}
  of $\Cousinhood$.
Thus, $w'$ and $w''$ belong to the same equivalence class of~$\Cousinhood$,
but this contradicts that $w'\in \AgentSubsetW$ and $w''\in \AgentSubsetW$ 
where $\AgentSubsetW'$ and $\AgentSubsetW''$ are included in two distinct equivalence
classes of~$\Cousinhood$. 
\medskip

  Now, we prove the converse assertion again by contradiction. 
  For this purpose, consider a 
  splitting $\AgentSubsetW=\AgentSubsetW'\sqcup \AgentSubsetW''$ such that
  $\TopologicalClosure{\AgentSubsetW}{\AgentSubsetW'}\cap\TopologicalClosure{\AgentSubsetW}{\AgentSubsetW''}=\emptyset$,
  and suppose that there exists $w' \in \AgentSubsetW'$ and 
  $w'' \in \AgentSubsetW''$ such that $w'$ and $w''$ are in the same equivalence classes of~$\Cousinhood$
  or otherwise said, such that $w' \Cousinhood w''$.
  Using Equation~\eqref{eq:Cousinhood} and the fact that
  $w'\not= w''$, as $\AgentSubsetW'\cap \AgentSubsetW''=\emptyset$,
  we deduce that
  $w' \bpTransitiveClosure{\Delta_{\AgentSubsetW} \ConditionalActiveTwo
    \Delta_{\AgentSubsetW} } w''$. Hence, there exists $k\in \NN$, $k\ge 1$, and
  a sequence $\sequence{w_i}{i\in \ic{1,k}}$ in~$\AgentSubsetW$ such that, for all $i\in \ic{1,k-1}$,
  we have $w_i \ConditionalActiveTwo w_{i+1}$ and $w' \ConditionalActiveTwo w_1$ and
  $w_k \ConditionalActiveTwo w''$. 
  Setting $w_0=w'$ and $w_{k+1}=w''$ and using the property that
$\bgent \ConditionalActiveTwo \cgent \implies 
\TopologicalClosure{\AgentSubsetW}{\na{\bgent}} \cap
  \TopologicalClosure{\AgentSubsetW}{\na{\cgent}}\not=\emptyset$
(shown at the beginning of the proof of Lemma~\ref{lem:TopologicalAw}), we get that
  $\TopologicalClosure{\AgentSubsetW}{\na{w_i}} \cap
  \TopologicalClosure{\AgentSubsetW}{\na{w_{i+1}}}\not=\emptyset$ 
for $i\in \ic{0,k}$. 
Now, the sequence $\sequence{w_i}{i\in \ic{0,k+1}}$ is in~$\AgentSubsetW$,
with the first element, $w'$, in $\AgentSubsetW'$ and the last one, $w''$, in
$\AgentSubsetW''$.
As $\AgentSubsetW = \AgentSubsetW'\sqcup \AgentSubsetW''$, 
  we can find two consecutive elements in the sequence such that 
  one is in $\AgentSubsetW'$ and the other one is in $\AgentSubsetW''$ and which are such that the intersection 
  of their topological closure is not empty. Thus, we have obtained that
  $\TopologicalClosure{\AgentSubsetW}{\AgentSubsetW'}\cap
  \TopologicalClosure{\AgentSubsetW}{\AgentSubsetW''}\not= \emptyset$, which
  gives a contradiction.
  
  This ends the proof.
\end{proof}

\section{Conclusion}

Together with its two companion
papers~\cite{Chancelier-De-Lara-Heymann-2021,Heymann-De-Lara-Chancelier-2021}, 
this paper is a contribution to providing
another perspective on conditional independence and do-calculus. 
In this paper, we consider directed graphs  (DGs), not necessarily acyclic, and we 
introduce a suitable topology on the set of vertices and 
the new notion of topological conditional separation on DGs. 
Then, we prove its equivalence with an extension of Pearl's d-separation on DGs.
What is more, we put forward a practical characterization of t-separation between subsets of
vertices. The proofs partially rely on results proven in~\cite{Chancelier-De-Lara-Heymann-2021}.

Checking topological separation is a two steps process. The first one, which is combinatorial,
consists in exploring the possible splitting of the conditioning set $\AgentSubsetW$ and the second one
consists in checking that the two closures induced by the splitting do not intersect.
It should be noted that, once given the splitting, the second step is computationaly easy and thus
the splitting appears as a ``certificate'' of conditional independence.
By contrast, checking d-separation is a one step combinatorial process as it requires to check that
\emph{all} the paths that connect two variables are blocked.

\appendix

\section{Additional material  on Alexandrov topology}
\label{Additional_results_alexandrov}

We use the material introduced
in~\S\ref{Topologies_induced_by_binary_relations},
and we provide additional results on Alexandrov topologies. 

\begin{proposition}
  Let  \( \EDGE_1,\EDGE_2 \) be two binary relations on the set~\( \VERTEX \).
  We have that
  \begin{align}
    \EDGE_1\subset\EDGE_2
    &\implies
      \TopologyLower{\EDGE_2} \subset \TopologyLower{\EDGE_1}
      \eqfinv
    \\
    \TopologyLower{\EDGE_1\cup\EDGE_2}
    &=
      \TopologyLower{\EDGE_1} \cap \TopologyLower{\EDGE_2}
      \label{eq:TopologyLower_EDGE_1cupEDGE_2}      
      \eqfinp             
  \end{align}
  The product topology \( \TopologyLower{\EDGE_1} \otimes \TopologyLower{\EDGE_2}
  \) on the product set~\( \VERTEX^2 \) coincides with
  the topology \(  \TopologyLower{\EDGE_1{\times}\EDGE_2} \)
  in~\eqref{eq:TopologyLowerEDGE},
  where \( \EDGE_1\times\EDGE_2 \subset \VERTEX^2\times\VERTEX^2 \)
  is the product binary relation on the product set~\( \VERTEX^2 \):
  \begin{equation}
    \TopologyLower{\EDGE_1} \otimes \TopologyLower{\EDGE_2}
    = \TopologyLower{\EDGE_1{\times}\EDGE_2} 
    \eqfinp
    \label{eq:product_topology}
  \end{equation}
  The topological closure of a subset~$\relation \subset \VERTEX^2$
  \wrt\ the topology \( \TopologyLower{\EDGE_1} \otimes \TopologyLower{\EDGE_2}
  = \TopologyLower{\EDGE_1{\times}\EDGE_2} \) is given by 
  \begin{equation}
    \TopologicalClosure{\EDGE_1{\times}\EDGE_2}{\relation} =
    \TransitiveReflexiveClosure{\EDGE_1}\relation
    \TransitiveReflexiveClosureConverse{\EDGE_2} 
    \eqfinp
    \label{eq:product_topology_topological_closure}
  \end{equation}
\end{proposition}

\begin{proof}
  The proof of~\eqref{eq:product_topology}
  relies on the following identity between open rectangles:
  \begin{align*}
    \Bgent_1\times\Bgent_2 \in \TopologyLower{\EDGE_1} \otimes \TopologyLower{\EDGE_2}
    &\iff
      \Bgent_1\in \TopologyLower{\EDGE_1} \text{ and }
      \Bgent_2 \in \TopologyLower{\EDGE_2}
\tag{by definition of the product topology \( \TopologyLower{\EDGE_1} \otimes
      \TopologyLower{\EDGE_2} \)}      
    \\
    &\iff
      \EDGE_1\Bgent_1 \subset \Bgent_1  \text{ and } 
      \EDGE_2\Bgent_2 \subset \Bgent_2
\tag{by~\eqref{eq:TopologyLowerEDGE_OpenSet}}
    \\
    &\iff
      \np{\EDGE_1\times\EDGE_2} \np{\Bgent_1\times\Bgent_2}
      \subset \Bgent_1\times\Bgent_2
      \eqfinp       
  \end{align*}

  Regarding~\eqref{eq:product_topology_topological_closure},
  by property~\eqref{eq:TopologicalClosureAlexandrovTopology} of an Alexandrov topology, we have that
  \begin{equation*}
    \TopologicalClosure{\EDGE_1{\times}\EDGE_2}{\relation} =
    \bigcup_{\np{\bgent', \cgent'} \in\relation}
    \TopologicalClosure{\EDGE_1{\times}\EDGE_2}{\na{\np{\bgent', \cgent'}}} 
  \end{equation*}
  so that, for any \( \bgent, \cgent \in \VERTEX \), we have 
  \begin{align*}
    \np{\bgent, \cgent} \in 
    \TopologicalClosure{\EDGE_1{\times}\EDGE_2}{\relation} 
    & \iff
      \exists \np{\bgent', \cgent'} \in\relation \eqsepv
      \np{\bgent, \cgent} \in 
      \TopologicalClosure{\EDGE_1{\times}\EDGE_2}{\na{\np{\bgent', \cgent'}}} 
    \\
    & \iff
      \exists \np{\bgent', \cgent'} \in\relation \eqsepv
      \np{\bgent, \cgent} \in 
      \TopologicalClosure{\EDGE_1}{\na{\bgent'}} \times
      \TopologicalClosure{\EDGE_2}{\na{\cgent'}}
\intertext{by property of the topological closure of rectangles in the product topology}      
    & \iff
      \exists \np{\bgent', \cgent'} \in\relation \eqsepv
      \bgent \in \TransitiveReflexiveClosure{\EDGE_1}\bgent'
      \eqsepv
      \cgent \in \TransitiveReflexiveClosure{\EDGE_2}\cgent'
\tag{by~\eqref{eq:TopologicalClosureEDGE}}
    \\
    & \iff
      \bgent \TransitiveReflexiveClosure{\EDGE_1}\relation
      \TransitiveReflexiveClosureConverse{\EDGE_2} \cgent 
  \end{align*}
  \medskip

  This ends the proof.   
\end{proof}

\begin{proposition}
  Let \( \npOrientedGraph \) be a graph,
  and $\Vertex \subset \VERTEX$ a subset.
  The following statements are equivalent:
  \begin{enumerate}
  \item The subset
    $\Vertex$ is a connected component of the topological space
    \( \np{\VERTEX, \TopologyLower{\EDGE}} \),
  \item The subset 
    $\Vertex$ is a connected component of the topological space
    \( \np{\VERTEX, \TopologyLower{\EDGE}'}=
    \np{\VERTEX, \TopologyLower{\Converse{\EDGE}}}  \),
  \item The subset
    $\Vertex$ is a connected component of the topological space
    \( \np{\VERTEX, \TopologyLower{\EDGE \cup \Converse{\EDGE}}}
    = \np{\VERTEX, \TopologyLower{\EDGE} \cap \TopologyLower{\Converse{\EDGE}}}
    = \np{\VERTEX, \TopologyLower{\EDGE} \cap \TopologyLower{\EDGE}'}
    \),
  \item The subset
    $\Vertex$ is an equivalence class of the equivalence relation
    \(  \npTransitiveReflexiveClosure{\EDGE \cup \Converse{\EDGE}} \). 
  \end{enumerate}
  \label{pr:connected_component}
\end{proposition}

\begin{proof}
  For any \( \vertex \in \VERTEX \), we denote by \( \hat\vertex \subset \VERTEX \) the
  $\TopologyLower{\EDGE}$-connected component of the topological space
  \( \np{\VERTEX, \TopologyLower{\EDGE}} \) that contains~$\vertex$, that is,
  the union of all $\TopologyLower{\EDGE}$-connected subsets of~$\VERTEX$
  that contain~$\vertex$, and we prove that
  \( \hat\vertex =\npTransitiveReflexiveClosure{\EDGE \cup
    \Converse{\EDGE}}\vertex \). For this purpose, we notice that,  
  by~\eqref{eq:TopologicalClosureEDGE}, we have that \( \npTransitiveReflexiveClosure{\EDGE \cup
    \Converse{\EDGE}}\vertex = \TopologicalClosure{\EDGE \cup
    \Converse{\EDGE}}{\vertex} \), 
  which is also the smallest $\TopologyLower{\EDGE}$-clopen set containing~$\vertex$
  by~\eqref{eq:TopologyLower_EDGE_1cupEDGE_2}.
  On the one hand, it is known, and can readily be shown, that \( \hat\vertex \subset \TopologicalClosure{\EDGE \cup
    \Converse{\EDGE}}{\vertex} \). 
  Indeed, if $\ClopenSet$ is any $\TopologyLower{\EDGE}$-clopen set containing~$\vertex$, then
  \( \hat\vertex \cap \ClopenSet \) and \( \hat\vertex \cap \Complementary{\ClopenSet} \) 
  are two disjoint $\TopologyLower{\EDGE}$-open sets whose union equal~\( \hat\vertex \).
  As this latter set is $\TopologyLower{\EDGE}$-connected, and as \( \vertex \in \hat\vertex \cap
  \ClopenSet \), we deduce that \( \hat\vertex \cap \ClopenSet = \hat\vertex \),
  hence that \( \hat\vertex \subset \ClopenSet \).
  On the other hand, it is well-known that the $\TopologyLower{\EDGE}$-connected component~\(
  \hat\vertex \) is closed. In an Alexandrov topology, \( \hat\vertex \) is also
  $\TopologyLower{\EDGE}$-open since the connected components form a partition of~$\VERTEX$, so that
  \( \Complementary{\hat\vertex} \) is a union of $\TopologyLower{\EDGE}$-closed sets, hence is
  closed. Therefore, \( \hat\vertex \) is a $\TopologyLower{\EDGE}$-clopen set, and we deduce that
  \( \TopologicalClosure{\EDGE \cup \Converse{\EDGE}}{\vertex}
  \subset \hat\vertex \).
  Therefore, we have obtained that \( \hat\vertex =
  \TopologicalClosure{\EDGE \cup \Converse{\EDGE}}{\vertex}=
  \npTransitiveReflexiveClosure{\EDGE \cup \Converse{\EDGE}}\vertex \). 

  Then, we easily deduce the equivalence between the four assertions.
  This ends the proof. 
\end{proof}
\section{Technical lemmas}
\label{Additional_Lemmas}

Here below, the relations~$\TopologicalRelation$ and $\ConverseTopologicalRelation$
have been introduced in~\eqref{eq:TopologicalRelation}
and \eqref{eq:ConverseTopologicalRelation}.
The following lemma is proved in~\cite{Chancelier-De-Lara-Heymann-2021}.

\begin{lemma}[\cite{Chancelier-De-Lara-Heymann-2021}]
  We have that
  \begin{subequations}
    \begin{align}
      \ConverseConditionalAncestor \ConditionalAncestor
      &=
        {\Delta 
        \cup
        \Delta_{\Complementary{\AgentSubsetW}} \ConditionalAscendent 
        \cup \ConverseConditionalAscendent \Delta_{\Complementary{\AgentSubsetW}}
        \cup \ConditionalCommonCause
        }
        \eqfinv
        \label{eq:AA}
      \\
      \Cousinhood \ConverseConditionalAncestor \ConditionalAncestor
      &=
        \Cousinhood \bp{
        \Delta \cup
        \ConverseConditionalAscendent
        \Delta_{{\Complementary{\AgentSubsetW}}}
        \cup \ConditionalCommonCause}
        \eqfinv
        \label{eq:CEmEsubset}
     \\
      \Cousinhood \ConverseConditionalAncestor \ConditionalAncestor \Cousinhood
      &= \Cousinhood
        \eqfinv
        \label{eq:gammaAAgamma}
      \\
      \Cousinhood \ConverseConditionalAncestor \ConditionalAncestor
      \Delta_{\Complementary{\AgentSubsetW}}
      &= \Cousinhood \bp{
        \ConverseConditionalAscendent 
        \cup \ConditionalCommonCause}
        \Delta_{\Complementary{\AgentSubsetW}}
        \eqfinv
        \label{eq:gammaAADelta}
      \\
      \Delta_{\Complementary{\AgentSubsetW}} \ConverseConditionalAncestor \ConditionalAncestor    \Cousinhood 
      &=
        \Delta_{\Complementary{\AgentSubsetW}}
        \bp{
        \ConditionalAscendent 
        \cup \ConditionalCommonCause}
        \Cousinhood
        \eqfinp
        \label{eq:DeltaAAgamma}
     \end{align}
  \end{subequations}
\end{lemma}

\begin{lemma}
  We have that
  \begin{equation}
    \begin{split}
  \Delta_{\Complementary{\AgentSubsetW}}
      \bp{ \Delta \cup \np{\ConditionalAscendent \cup \ConverseConditionalActiveTwo}\Cousinhood}
      \ConverseConditionalAncestor\ConditionalAncestor \Cousinhood \ConverseConditionalAncestor\ConditionalAncestor 
      \bp{ \Delta \cup\Cousinhood \np{\ConverseConditionalAscendent \cup \ConverseConditionalActiveTwo}}
      \Delta_{\Complementary{\AgentSubsetW}}
      \\ =
     \Delta_{\Complementary{\AgentSubsetW}}
      { \np{\ConditionalAscendent \cup \ConverseConditionalActiveTwo}\Cousinhood}
      { \np{\ConverseConditionalAscendent \cup \ConverseConditionalActiveTwo}}
      \Delta_{\Complementary{\AgentSubsetW}}
    \end{split}
         \label{eq:gamma-rel-lambda-1}
  \end{equation}
\end{lemma}

\begin{proof}
 We have that 
  \begin{align*}
    \Delta_{\Complementary{\AgentSubsetW}}
    &
      \bp{ \Delta \cup \np{\ConditionalAscendent \cup \ConverseConditionalActiveTwo}\Cousinhood}
      \ConverseConditionalAncestor\ConditionalAncestor \Cousinhood \ConverseConditionalAncestor\ConditionalAncestor 
      \bp{ \Delta \cup\Cousinhood \np{\ConverseConditionalAscendent \cup \ConverseConditionalActiveTwo}}
      \Delta_{\Complementary{\AgentSubsetW}}
      \nonumber 
    \\
    &= 
      \Delta_{\Complementary{\AgentSubsetW}}
      \ConverseConditionalAncestor\ConditionalAncestor \Cousinhood \ConverseConditionalAncestor\ConditionalAncestor 
      \Delta_{\Complementary{\AgentSubsetW}}
      \nonumber \tag{by developing}
    \\
    &\hspace{0.5cm}\cup 
      \Delta_{\Complementary{\AgentSubsetW}}
      \ConverseConditionalAncestor\ConditionalAncestor \Cousinhood \ConverseConditionalAncestor\ConditionalAncestor 
      \bp{ \Cousinhood \np{\ConverseConditionalAscendent \cup \ConverseConditionalActiveTwo}}
      \Delta_{\Complementary{\AgentSubsetW}}
      \nonumber 
    \\
    &\hspace{0.5cm}\cup
      \Delta_{\Complementary{\AgentSubsetW}}
      \bp{ \np{\ConditionalAscendent \cup \ConverseConditionalActiveTwo}\Cousinhood}
      \ConverseConditionalAncestor\ConditionalAncestor \Cousinhood \ConverseConditionalAncestor\ConditionalAncestor 
      \Delta_{\Complementary{\AgentSubsetW}}
      \nonumber 
    \\
    &\hspace{0.5cm}\cup
      \Delta_{\Complementary{\AgentSubsetW}}
      \bp{ \np{\ConditionalAscendent \cup \ConverseConditionalActiveTwo}\Cousinhood}
      \ConverseConditionalAncestor\ConditionalAncestor \Cousinhood \ConverseConditionalAncestor\ConditionalAncestor 
      \bp{ \Cousinhood \np{\ConverseConditionalAscendent \cup \ConverseConditionalActiveTwo}}
      \Delta_{\Complementary{\AgentSubsetW}}
      \nonumber 
    \\
    &= 
      \Delta_{\Complementary{\AgentSubsetW}}
      \ConverseConditionalAncestor\ConditionalAncestor \Cousinhood \ConverseConditionalAncestor\ConditionalAncestor 
      \Delta_{\Complementary{\AgentSubsetW}}
      \nonumber 
    \\
    &\hspace{0.5cm}\cup 
      \Delta_{\Complementary{\AgentSubsetW}}
      \ConverseConditionalAncestor\ConditionalAncestor \Cousinhood 
      { \np{\ConverseConditionalAscendent \cup \ConverseConditionalActiveTwo}}
      \Delta_{\Complementary{\AgentSubsetW}}
      \tag{
      as $\Cousinhood \ConverseConditionalAncestor \ConditionalAncestor \Cousinhood = \Cousinhood$
      by~\eqref{eq:gammaAAgamma}
      }
    \\
    &\hspace{0.5cm}\cup
      \Delta_{\Complementary{\AgentSubsetW}}
      { \np{\ConditionalAscendent \cup \ConverseConditionalActiveTwo}\Cousinhood}
      \ConverseConditionalAncestor\ConditionalAncestor 
      \Delta_{\Complementary{\AgentSubsetW}}
      \tag{also by~\eqref{eq:gammaAAgamma}
      }
    \\
    &\hspace{0.5cm}\cup
      \Delta_{\Complementary{\AgentSubsetW}}
      { \np{\ConditionalAscendent \cup \ConverseConditionalActiveTwo}\Cousinhood}
      { \np{\ConverseConditionalAscendent \cup \ConverseConditionalActiveTwo}}
      \Delta_{\Complementary{\AgentSubsetW}}
      \tag{also by~\eqref{eq:gammaAAgamma} applied twice}
    \\
    &= 
      \Delta_{\Complementary{\AgentSubsetW}}\bp{\ConditionalAscendent \cup \ConditionalActiveTwo}
      \Cousinhood \bp{\ConverseConditionalAscendent \cup \ConditionalActiveTwo} \Delta_{\Complementary{\AgentSubsetW}}
      \tag{by~\eqref{eq:gammaAADelta} and~\eqref{eq:DeltaAAgamma}}
    \\
    &\hspace{0.5cm}\cup 
      \Delta_{\Complementary{\AgentSubsetW}}\bp{\ConditionalAscendent \cup \ConditionalActiveTwo}
      { \np{\ConverseConditionalAscendent \cup \ConverseConditionalActiveTwo}}
      \Delta_{\Complementary{\AgentSubsetW}}
      \tag{by~\eqref{eq:DeltaAAgamma}}
    \\
    &\hspace{0.5cm}\cup
      \Delta_{\Complementary{\AgentSubsetW}}
      \np{\ConditionalAscendent \cup \ConverseConditionalActiveTwo}
      \Cousinhood \bp{\ConverseConditionalAscendent \cup \ConditionalActiveTwo} \Delta_{\Complementary{\AgentSubsetW}}
      \tag{by~\eqref{eq:gammaAADelta}}
    \\
    &\hspace{0.5cm}\cup
      \Delta_{\Complementary{\AgentSubsetW}}
      { \np{\ConditionalAscendent \cup \ConverseConditionalActiveTwo}\Cousinhood}
      { \np{\ConverseConditionalAscendent \cup \ConverseConditionalActiveTwo}}
      \Delta_{\Complementary{\AgentSubsetW}}
      \nonumber
    \\
    &= 
      \Delta_{\Complementary{\AgentSubsetW}}
      { \np{\ConditionalAscendent \cup \ConverseConditionalActiveTwo}\Cousinhood}
      { \np{\ConverseConditionalAscendent \cup \ConverseConditionalActiveTwo}}
      \Delta_{\Complementary{\AgentSubsetW}}
       \eqfinp
  \end{align*}
  
  This ends the proof.
\end{proof}

\begin{lemma}
  \label{le:equation_for_topological_equivalence}      
  We have that
  \begin{equation}
    \Delta_{\Complementary{\AgentSubsetW}}
    \ConverseTopologicalRelation\ConverseConditionalAncestor \ConditionalAncestor
    \TopologicalRelation\Delta_{\Complementary{\AgentSubsetW}}
    =
    \Delta_{\Complementary{\AgentSubsetW}}
    \ConditionalActiveNew
    \Delta_{\Complementary{\AgentSubsetW}}
    \eqfinp
    \label{eq:equation_for_topological_equivalence}      
  \end{equation}
\end{lemma}

\begin{proof}
  First, we write 
  \begin{align*}
    \ConverseTopologicalRelation
    &\ConverseConditionalAncestor \ConditionalAncestor
      \TopologicalRelation
    \\
    &=
      \underbrace{\Bp{\Delta \cup \bp{\ConditionalAscendent \cup
      \ConditionalActiveTwo}\Cousinhood}}_{\ConverseTopologicalRelation
      \textrm{~by~\eqref{eq:ConverseTopologicalRelation}}}
      \ConverseConditionalAncestor \ConditionalAncestor
      \underbrace{\Bp{\Delta \cup \Cousinhood\bp{\ConverseConditionalAscendent \cup \ConverseConditionalActiveTwo}}}
      _{\TopologicalRelation       \textrm{~by~\eqref{eq:ConverseTopologicalRelation}}}
    \\
    &=\np{\ConverseConditionalAncestor \ConditionalAncestor}
      \cup
      \Bp{\bp{\ConditionalAscendent \cup \ConditionalActiveTwo}\Cousinhood
      \ConverseConditionalAncestor \ConditionalAncestor}
      \cup
      \Bp{
      \ConverseConditionalAncestor \ConditionalAncestor
      \Cousinhood\bp{\ConverseConditionalAscendent \cup \ConverseConditionalActiveTwo}
      }
      \nonumber
    \\
    &\hspace{1cm}
      \cup
      \Bp{
      \bp{\ConditionalAscendent \cup \ConditionalActiveTwo}
      \underbrace{\Cousinhood\ConverseConditionalAncestor\ConditionalAncestor\Cousinhood}%
      _{= \Cousinhood \textrm{ by~\eqref{eq:gammaAAgamma}}}
      \bp{\ConverseConditionalAscendent \cup \ConverseConditionalActiveTwo}
      }
      \tag{by developing} 
    \\
    &=\np{\ConverseConditionalAncestor \ConditionalAncestor}
      \cup
      \Bp{\bp{\ConditionalAscendent \cup \ConditionalActiveTwo}\Cousinhood
      \ConverseConditionalAncestor \ConditionalAncestor}
      \cup
      \Bp{
      \ConverseConditionalAncestor \ConditionalAncestor
      \Cousinhood\bp{\ConverseConditionalAscendent \cup \ConverseConditionalActiveTwo}
      }
      \nonumber
    \\
    &\hspace{1cm}
      \cup \ConditionalActiveThree
      \eqfinp
  \end{align*}
  Second, we obtain that 
  \begin{align*}
    \Delta_{\Complementary{\AgentSubsetW}}
    &
      \ConverseTopologicalRelation\ConverseConditionalAncestor \ConditionalAncestor  \TopologicalRelation
      \Delta_{\Complementary{\AgentSubsetW}}
      \nonumber
    \\
    &=\Bp{\Delta_{\Complementary{\AgentSubsetW}}\np{\ConverseConditionalAncestor \ConditionalAncestor}\Delta_{\Complementary{\AgentSubsetW}}}
      \cup
      \Bp{
      \Delta_{\Complementary{\AgentSubsetW}}
      \bp{\ConditionalAscendent \cup \ConditionalActiveTwo}
      \underbrace{\Cousinhood
      \ConverseConditionalAncestor \ConditionalAncestor\Delta_{\Complementary{\AgentSubsetW}}}_{=\Cousinhood \np{
      \ConverseConditionalAscendent 
      \cup \ConditionalActiveTwo}
      \Delta_{\Complementary{\AgentSubsetW}} \text{ by~\eqref{eq:gammaAADelta}}
      }}
      \nonumber 
    \\
    &\hspace{1cm}
      \cup
      \Bp{\underbrace{\Delta_{\Complementary{\AgentSubsetW}}
      \ConverseConditionalAncestor \ConditionalAncestor
      \Cousinhood
      }_{ =
      \Delta_{\Complementary{\AgentSubsetW}}
      \np{
      \ConditionalAscendent 
      \cup \ConditionalActiveTwo}
      \Cousinhood \text{ by~\eqref{eq:DeltaAAgamma}}
      }
      \bp{\ConverseConditionalAscendent \cup \ConverseConditionalActiveTwo}
      \Delta_{\Complementary{\AgentSubsetW}}    }\nonumber 
    \\
    &\hspace{1cm}
      \cup \Bp{\Delta_{\Complementary{\AgentSubsetW}} \ConditionalActiveThree\Delta_{\Complementary{\AgentSubsetW}}}
      \nonumber 
    \\
    &=
      \bp{\Delta_{\Complementary{\AgentSubsetW}}
      \np{\ConverseConditionalAncestor \ConditionalAncestor}
      \Delta_{\Complementary{\AgentSubsetW}}}
      \cup
      \bgp{\Delta_{\Complementary{\AgentSubsetW}}
      \ConditionalActiveThree
      \Delta_{\Complementary{\AgentSubsetW}}}
      \intertext{because the three last terms in the union are all equal}      
      \nonumber 
    &=
      \bp{\Delta_{\Complementary{\AgentSubsetW}}
      \np{
      \Delta 
      \cup
      \Delta_{\Complementary{\AgentSubsetW}} \ConditionalAscendent 
      \cup \ConverseConditionalAscendent \Delta_{\Complementary{\AgentSubsetW}}
      \cup \ConditionalActiveTwo
      }
      \Delta_{\Complementary{\AgentSubsetW}}}
      \nonumber
    \\
    &\hspace{1cm}
      \cup
      \bgp{\Delta_{\Complementary{\AgentSubsetW}}
      \ConditionalActiveThree
      \Delta_{\Complementary{\AgentSubsetW}}}
      \tag{by~\eqref{eq:AA}}
    \\
    &=
      \Delta_{\Complementary{\AgentSubsetW}}
      \bgp{
      \Delta 
      \cup \ConditionalAscendent 
      \cup \ConverseConditionalAscendent 
      \cup \ConditionalActiveTwo
      \cup
      \ConditionalActiveThree
      }
      \Delta_{\Complementary{\AgentSubsetW}}
    \\
    &=
      \Delta_{\Complementary{\AgentSubsetW}}
      \ConditionalActiveNew
      \Delta_{\Complementary{\AgentSubsetW}}
      \eqfinp
      \tag{by definition of $\ConditionalActive$ in~\eqref{eq:conditional_active_relation}}
  \end{align*}
  
  This ends the proof.
\end{proof}

\ifpreprint
\preprintstart 
In Lemma~\ref{le:equation_for_topological_equivalence}, it was proved that the two relations 
$\ConverseTopologicalRelation\ConverseConditionalAncestor \ConditionalAncestor\TopologicalRelation$ and 
$\ConditionalActiveNew$ coincide when restricted to the subset $\Complementary{\AgentSubsetW}$. 
More generally, we give in this last lemma the relationship between these two relations. 

\begin{lemma}
  We have that 
  \begin{equation}
    \bp{\ConverseTopologicalRelation\ConverseConditionalAncestor \ConditionalAncestor  \TopologicalRelation}
    \cup
    \ConverseTopologicalRelation 
    \ConverseConditionalAscendent \Delta_{\AgentSubsetW}
    \cup
    \Delta_{\AgentSubsetW}\ConditionalAscendent
    \TopologicalRelation 
    =
    \ConditionalActiveNew  \cup  \ConverseTopologicalRelation \cup \TopologicalRelation
    \eqfinp     
  \end{equation}
\end{lemma}

\begin{proof}
  We use the notation
  $\RelTheta=\ConditionalAscendent \cup \ConditionalActiveTwo$ and
  $\ConverseRelTheta = \ConverseConditionalAscendent \cup \ConditionalActiveTwo = \npConverse{\RelTheta}$
  to simplify the reading of the proof, so that
  \begin{align*}
    \ConverseTopologicalRelation
    & \ConverseConditionalAncestor\ConditionalAncestor\TopologicalRelation
    \\
    &=
      \np{\Delta \cup
      \bp{\ConditionalAscendent \cup \ConditionalActiveTwo}\Cousinhood}
      \ConverseConditionalAncestor\ConditionalAncestor
      \np{\Delta \cup \Cousinhood\bp{\ConverseConditionalAscendent \cup \ConverseConditionalActiveTwo}}
      \tag{by definition of the relations~$\TopologicalRelation$ and $\ConverseTopologicalRelation$
      in~\eqref{eq:TopologicalRelation} and \eqref{eq:ConverseTopologicalRelation}}
    \\
    &=
      \np{\Delta \cup \RelTheta\Cousinhood}
      \ConverseConditionalAncestor\ConditionalAncestor
      \np{\Delta \cup \Cousinhood\ConverseRelTheta}
      \tag{using the just defined $\RelTheta=\ConditionalAscendent \cup \ConditionalActiveTwo$
      and 
      $\ConverseRelTheta = \ConverseConditionalAscendent \cup \ConditionalActiveTwo$}
    \\
    &=
      \np{\ConverseConditionalAncestor \ConditionalAncestor}
      \cup
      \np{\RelTheta\Cousinhood
      \ConverseConditionalAncestor \ConditionalAncestor}
      \cup
      \ConditionalActiveThreeTheta
      \cup
      \np{
      \ConverseConditionalAncestor \ConditionalAncestor
      \Cousinhood\ConverseRelTheta
      }
      \tag{by developing and by~\eqref{eq:gammaAAgamma} giving
      $\Cousinhood \ConverseConditionalAncestor \ConditionalAncestor \Cousinhood= \Cousinhood$
      }
    \\
    &=
      \np{\ConverseConditionalAncestor \ConditionalAncestor}
      \cup
      {\RelTheta\Cousinhood}
      \bp{\Delta 
      \cup \ConverseConditionalAscendent
      \Delta_{\Complementary{\AgentSubsetW}}
      \cup \ConditionalActiveTwo}
      \cup
      \ConditionalActiveThreeTheta
      \cup
      \bp{\Delta 
      \cup \Delta_{\Complementary{\AgentSubsetW}} \ConditionalAscendent
      \cup \ConditionalActiveTwo }
      {\Cousinhood\ConverseRelTheta}
  \end{align*}
  as by~\eqref{eq:CEmEsubset} ${\RelTheta\Cousinhood\ConverseConditionalAncestor \ConditionalAncestor}
  = {\RelTheta\Cousinhood}
  \bp{\Delta 
    \cup \ConverseConditionalAscendent
    \Delta_{\Complementary{\AgentSubsetW}}
    \cup \ConditionalActiveTwo}$ and by symmetry for the last term. Thus, 
  using the last equality and performing a union with
  ${\RelTheta\Cousinhood} \ConverseConditionalAscendent \Delta_{\AgentSubsetW}
  \cup \Delta_{\AgentSubsetW}\ConditionalAscendent {\Cousinhood\ConverseRelTheta}$
  on both sides of the equality, we obtain
  \begin{align*}
    &\bp{\ConverseTopologicalRelation\ConverseConditionalAncestor \ConditionalAncestor  \TopologicalRelation}
      \cup {\RelTheta\Cousinhood} \ConverseConditionalAscendent \Delta_{\AgentSubsetW}
      \cup \Delta_{\AgentSubsetW}\ConditionalAscendent {\Cousinhood\ConverseRelTheta} 
    \\
    &\hspace{1cm}=
      \np{\ConverseConditionalAncestor \ConditionalAncestor}
      \cup
      {\RelTheta\Cousinhood}
      \bp{\Delta 
      \cup \ConverseConditionalAscendent
      \cup \ConditionalActiveTwo}
      \cup
      \ConditionalActiveThreeTheta
      \cup
      \bp{\Delta 
      \cup \ConditionalAscendent
      \cup \ConditionalActiveTwo }
      {\Cousinhood\ConverseRelTheta}
    \\
    &\hspace{1cm}=
      \np{\ConverseConditionalAncestor \ConditionalAncestor}
      \cup
      {\RelTheta\Cousinhood}
      \bp{\Delta \cup \ConverseRelTheta}
      \cup
      \ConditionalActiveThreeTheta
      \cup
      \bp{\Delta \cup \RelTheta}
      {\Cousinhood\ConverseRelTheta}
    \\
    &\hspace{1cm}=
      \np{\ConverseConditionalAncestor \ConditionalAncestor}
      \cup
      {\RelTheta\Cousinhood}
      \cup
      \ConditionalActiveThreeTheta
      \cup
      {\Cousinhood\ConverseRelTheta}
    \\
    &\hspace{1cm}=
      \np{\Delta
      \cup
      \Delta_{\Complementary{\AgentSubsetW}}     \ConditionalAscendent 
      \cup \ConverseConditionalAscendent\Delta_{\Complementary{\AgentSubsetW}}
      \cup  \ConditionalActiveTwo
      }
      \cup
      {\RelTheta\Cousinhood}
      \cup
      \ConditionalActiveThreeTheta
      \cup
      {\Cousinhood\ConverseRelTheta}
      \tag{by~\eqref{eq:AA}}
    \\
    &\hspace{1cm}=
      \np{\Delta
      \cup
      \Delta_{\Complementary{\AgentSubsetW}}     \ConditionalAscendent 
      \cup \ConverseConditionalAscendent\Delta_{\Complementary{\AgentSubsetW}}
      \cup  \ConditionalActiveTwo
      }
      \cup
      \underbrace{\np{\Delta \cup {\RelTheta\Cousinhood}}}_{=\ConverseTopologicalRelation}
      \cup
      \ConditionalActiveThreeTheta
      \cup
      \underbrace{\np{\Delta \cup{\Cousinhood\ConverseRelTheta}}}_{=\TopologicalRelation}
    \\
    &\hspace{1cm}=
      \Delta
      \cup
      \Delta_{\Complementary{\AgentSubsetW}}     \ConditionalAscendent 
      \cup \ConverseConditionalAscendent\Delta_{\Complementary{\AgentSubsetW}}
      \cup  \ConditionalActiveTwo
      \cup
      \ConditionalActiveThreeTheta
      \cup
      \ConverseTopologicalRelation
      \cup
      \TopologicalRelation
      \eqfinp
  \end{align*}
  Finaly,
  using the last equality and performing a union with
  $\Delta_{\AgentSubsetW}\ConditionalAscendent
  \cup \ConverseConditionalAscendent \Delta_{\AgentSubsetW}$
  on both sides of the equality, we obtain
  \begin{align*}
    &\bp{\ConverseTopologicalRelation\ConverseConditionalAncestor \ConditionalAncestor  \TopologicalRelation}
      \cup {\RelTheta\Cousinhood} \ConverseConditionalAscendent \Delta_{\AgentSubsetW}
      \cup \Delta_{\AgentSubsetW}\ConditionalAscendent {\Cousinhood\ConverseRelTheta}
      \cup \Delta_{\AgentSubsetW}\ConditionalAscendent
      \cup \ConverseConditionalAscendent \Delta_{\AgentSubsetW}
    \\
    &\hspace{1cm}=
      \Delta
      \cup
      \ConditionalAscendent 
      \cup \ConverseConditionalAscendent
      \cup  \ConditionalActiveTwo
      \cup
      \ConditionalActiveThreeTheta
      \cup
      \ConverseTopologicalRelation
      \cup
      \TopologicalRelation
    \\
    &\hspace{1cm}=
      \Delta
      \cup
      \ConditionalAscendent 
      \cup \ConverseConditionalAscendent
      \cup  \ConditionalActiveTwo
      \cup
      \ConditionalActiveThreeTheta
      \cup
      \ConverseTopologicalRelation
      \cup
      \TopologicalRelation
    \\
    &\hspace{1cm}=
      \ConditionalActiveNew
      \cup
      \ConverseTopologicalRelation
      \cup
      \TopologicalRelation
      \tag{by definition of $\ConditionalActive$ in~\eqref{eq:conditional_active_relation}}
      \eqfinp
  \end{align*}
  
  This ends the proof.
\end{proof}
\preprintstop
\fi

\newcommand{\noopsort}[1]{} \ifx\undefined\allcaps\def\allcaps#1{#1}\fi

\end{document}